\documentclass{article}

\usepackage{arxiv}

\usepackage{amsmath}
\usepackage{amsthm}
\usepackage{amssymb}

\theoremstyle{definition}
\newtheorem{definition}{Definition}
\newtheorem{proposition}[definition]{Proposition}
\newtheorem{lemma}[definition]{Lemma}

\newtheorem{theorem}[definition]{Theorem}
\newtheorem{claim}{Claim}

\usepackage[utf8]{inputenc} 
\usepackage[T1]{fontenc}    
\usepackage{hyperref}       
\usepackage{url}            
\usepackage{booktabs}       
\usepackage{amsfonts}       
\usepackage{nicefrac}       
\usepackage{microtype}      
\usepackage{cleveref}       
\usepackage{lipsum}         
\usepackage{graphicx}
\usepackage{natbib}
\usepackage{doi}

\usepackage{hyperref}
\usepackage{algorithmic}
\usepackage{graphicx}
\usepackage{textcomp}
\usepackage{xcolor}
\usepackage{tabularx}
\usepackage{tikz}
\usetikzlibrary{arrows}
\usetikzlibrary{matrix}
\usepackage{bm}
\newcommand{\setof}[2]{\{{#1} \mid {#2}\}}
\newcommand{\arrow}[1]{\stackrel{#1}{\longrightarrow}}
\newcommand{\dobarrow}[1]{\stackrel{#1}{\Longrightarrow}}

\title{A New Branching Bisimulation for Probabilistic Processes}

\newif\ifuniqueAffiliation
\uniqueAffiliationfalse

\ifuniqueAffiliation 
}
\else
\usepackage{authblk}

\newbox{\orcid}\sbox{\orcid}{\includegraphics[scale=0.06]{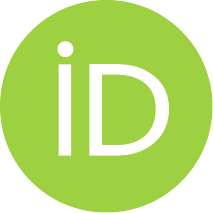}} 
\author[1,2]{%
	\href{https://orcid.org/0009-0007-7588-1964}{\usebox{\orcid}\hspace{1mm}Guo Li\thanks{\texttt{liguo2019@email.swu.edu.cn}}}%
}
\author[1]{%
	\href{https://orcid.org/0000-0002-5350-956X}{\usebox{\orcid}\hspace{1mm}Zhaokai Li\thanks{\texttt{lzk15892683467@email.swu.edu.cn}}}%
}
\author[3]{%
	\href{https://orcid.org/0000-0002-8334-8277}{\usebox{\orcid}\hspace{1mm}Xinxin Liu\thanks{\texttt{xinxin@ios.ac.cn}}}%
}
\author[1]{%
	\href{https://orcid.org/0000-0001-9771-3071}{\usebox{\orcid}\hspace{1mm}Zhiming Liu\thanks{\texttt{zhimingliu88@swu.edu.cn}}}%
}
\author[4]{%
	\href{https://orcid.org/0000-0002-4558-7867}{\usebox{\orcid}\hspace{1mm}Quan Sun\thanks{\texttt{quansun@hbmzu.edu.cn}}}%
}
\author[4]{%
	\href{https://orcid.org/0000-0001-8350-0788}{\usebox{\orcid}\hspace{1mm}Wei Zhang\thanks{\texttt{weiz@hbmzu.edu.cn}}}%
}
\affil[1]{College of Computer and Information Science, Southwest University, Chongqing, China}
\affil[2]{Hunan University of Arts and Science,Changde,China}
\affil[3]{Institute of Software, Chinese Academy of Sciences, Beijing, China}
\affil[4]{School of Mathematics and Statistics, Hubei Minzu University, Enshi, China}
\fi

\renewcommand{\shorttitle}{\textit{arXiv} Template}

\hypersetup{
pdftitle={A New Branching Bisimulation for Probabilistic Processes},
pdfsubject={cs.LO},
pdfauthor={Guo Li, Zhaokai Li},
pdfkeywords={Probabilistic processes, Probabilistic bisimilarity, Process calculi, Up-to technique, Congruence},
}

\begin{document}
\maketitle

\begin{abstract}
	We introduce a new branching bisimulation for probabilistic processes, which induces a more refined equivalence relation
	than any known equivalence that abstracts from unobservable transitions, with a rooted version that is a congruence for a 
	language of probabilistic process with the usual static as well as dynamic constructs including recursion.
\end{abstract}

\keywords{Probabilistic processes \and Probabilistic bisimilarity \and Process calculi \and Up-to technique \and Congruence}

\section{Introduction}

Process algebras, or process calculi, provide a rigorous framework for the modeling and analysis of concurrent systems~\cite{Baeten05,BaetenS14}. 
These calculi have been used in the verification of communication protocols, hardware circuits, distributed systems, security protocols, and web services, 
among others~\cite{GaravelL22}. In the development of process calculi, the idea of bisimulation, first introduced in  Robin Milner's Calculus of Communicating 
Systems (CCS)~\cite{milner89} has been a great milestone; it not only can be used to define various behavioral equivalences
but also provides an effective technique for proving process equivalences, known
as the Bisimulation Proof Method. 
 
Classical process calculi have focused on characterizing and verifying the functional aspects of concurrent systems. 
In the late 1980s, research on probabilistic extensions of process calculi emerged, which introduced quantitative analysis of concurrent system behaviors 
by augmenting non-deterministic actions of processes with probabilistic information.  These probabilistic process calculi make possible the analysis 
of non-functional aspects of concurrent processes, such as performance and reliability.

In the development of probabilistic process calculi, well-established notions of bisimulation for classic processes are adjusted to compare 
probabilistic processes, among which branching bisimulation~\cite{GlabbeekW96} is a popular one. 
Here we propose a new notion of branching bisimulation for probabilistic processes, which is more refined than any notion of bisimulation 
for probabilistic processes proposed in the literature that abstracts from internal or unobservable transitions.
We use the following example, inspired by an example in~\cite{Jonsson01}, to motivate the idea. 
\begin{figure}[ht]
\centering

\includegraphics[width=\linewidth]{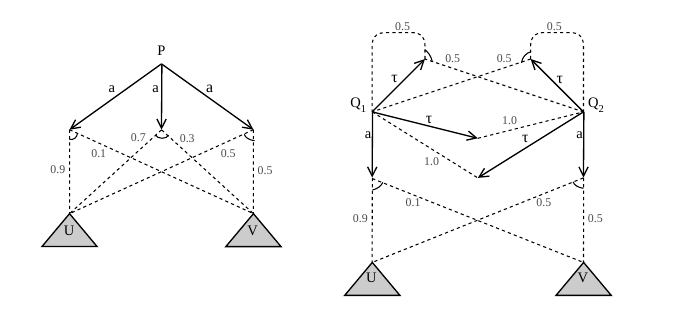}
\caption{Example 1}
\label{fig:prob_transition}
\end{figure}
The left part of Figure 1  shows a process which
starts from initial state $P$ can perform three $a$-actions. For example, the $a$-action in the middle can take the process to some state which is $U$ with probability 0.7, or is $V$ with
probability $0.3$ (assuming that $U$ and $V$ behave differently). 
The right part shows a process which starts from the state $Q_1$ can perform two internal actions labeled $\tau$, one of which 
can take the process to the state $Q_2$ with probability 1.0, the other to the states $Q_1$ and $Q_2$ both with probability 0.5, etc.
Using a simple language of probabilistic process expressions, we can specify the behavior of the left and right processes by the following equations:
$$
\begin{array}{lcl}
P&=&a.(\langle 0.9\rangle U\uplus\langle0.1\rangle V)+a.(\langle 0.7\rangle U\uplus\langle 0.3\rangle V)+a.(\langle 0.5\rangle U\uplus\langle 0.5\rangle V)\\
Q_1&=&a.(\langle 0.9\rangle U\uplus\langle0.1\rangle V)+\tau.(\langle 1.0\rangle Q_2)+\tau.(\langle 0.5\rangle Q_1\uplus\langle 0.5\rangle Q_2)\\
Q_2&= &
a.(\langle 0.5\rangle U\uplus\langle0.5\rangle V)+\tau.(\langle 1.0\rangle Q_1)+\tau.(\langle 0.5\rangle Q_1\uplus\langle 0.5\rangle Q_2)
\end{array}
$$
To the best of our knowledge, all the notions of bisimulation in the literature  which take an abstract view of the internal action 
$\tau$ would consider $P,Q_1,Q_2$ all
equivalently behaving  states. 
Take the branching bisimulation studied in~\cite{CastiglioniT20} for example, to match the above-mentioned middle $a$ transition from $P$,
even though the state $Q_1$ cannot immediately perform an $a$ transition which ends in the states $U,V$ with probability $0.7$ and $0.3$ respectively, 
$Q_1$ can choose to perform the $\tau$ transition which ends in the states $Q_1$ and $Q_2$ in probability half and half respectively, and then
$Q_1$ and $Q_2$ can both perform their respective $a$ transitions, and the resulting {\em combined} effect of the two $a$ transitions 
could lead to a state which is either $U$ or $V$ with probability $0.7$ and $0.3$ respectively, thus matching the middle $a$ transition from $P$.
Although from a probabilistic perspective  such treatment is reasonable, there are also reasons to doubt it. In particular, here we are in fact
comparing the result of a normal transition to the result obtained by combining results of two transitions of two states.
Arguably,   using the result of a single transition to compare instead of that combined from two transitions 
is  preferred here,
particularly when our model is such that actions are from individual states not from a collection of those. Thus said, there seems to be the following
immediate difficulty. One might agree that in the case of strong bisimulation where only single transitions are considered
in each matching step, in the probabilistic case it is still both possible either
\begin{itemize}
\item to insist on comparing single transitions from individual states, which leads to the notion of {\em strong bisimulation} for probabilistic processes, or
\item  to allow comparing combined transitions (from a single state) with single transitions, which leads to the notion of {\em strong probabilistic bisimulation} for probabilistic 
processes.
\end{itemize}
Now, how 
could it be possible to only compare the results of single transition in the situation where a sequence of $\tau$ transitions can 
be allowed which obviously leads to 
a set of possible states.  In this paper we achieve this by generalizing the idea of inert transitions ($\tau$ transitions
within the same equivalence class)  from the study of classical branching bisimulation,
leading to a new notion of branching bisimulation under which the above states $P,Q_1,Q_2$ are pairwise different. For the moment we shall refrain from
discussing advantages and disadvantages of different notions, since a good virtue of the study of process calculi is a wealth of different and related
notions, each of which could find its place in applications. We also establish that with the addition of  the usual rootedness condition, 
our new notion of  branching bisimulation is a 
congruence for a language of probabilistic processes, extended from CCS with the usual static as well as  dynamic 
constructs including recursion.

The rest of the paper is organized as follows. Section~\ref{sec2} presents the extended language for probabilistic processes and its operational semantics. 
Section~\ref{sec4} introduces branching bisimilarity and proves that it is an equivalence relation. 
Section~\ref{sec5} first proves that branching bisimilarity is a congruence for parallel composition and other operators, and then introduces the notion of branching equality, proving that branching equality is a congruence for all operators. 
Section~\ref{sec8} provides a brief overview of some related work. Finally, Section~\ref{sec9} concludes.  

\section{A Language for Probabilistic Processes}\label{sec2}

We extend the process description language of Milner's
CCS \cite{milner89} to describe the probabilistic processes.
We assume an infinite set $\mathcal{A}$ of \emph{names} and use $a, b, c,\ldots$ to range over
$\mathcal{A}$. We denote by $\overline{\mathcal{A}}$  the set of \emph{co-names,}
i.e. $\overline{\mathcal{A}}
=\{\overline a \mid a\in\mathcal{A}\}$. Then we define $\mathcal{L}=\mathcal{A}\cup\overline{\mathcal{A}}$,
and call $\mathcal{L}$ the set of \emph{labels.}
We shall use $l$, $l'$ to range over $\mathcal{L}$.
Define $\mathit{Act}=\mathcal{L}\cup\{\tau\}$ to be the set of actions, where $\tau$ is the invisible action.
We shall use $\alpha$, $\beta$ to range over $\mathit{Act}$,
and $K$, $L$ to stand for subsets of $\mathcal{L}$.
For ease of notation we also treat $\overline{\,\cdot\,\vphantom{\alpha}}$
as a one-to-one function on $\mathit{Act}$ such that $\overline{\tau}=\tau$ and
$\overline{\overline{\alpha}}=\alpha$.
A \emph{relabelling function} $f$ is a function from $\mathcal{L}$ to $\mathcal{L}$ such that $f(\overline l)=\overline{f(l)}$,
and for ease of notation we extend $f$ to $\mathit{Act}$ by letting $f(\tau)=\tau$.
We also assume a set ${\mathcal K}$ of {\em constants} and
a set ${\mathcal X}$ of {\em process variables}; we use $A,B,C,D,\dots$ with possible
subscripts to range over ${\mathcal K}$
and $X,Y,\ldots$ with possible subscripts to range over ${\mathcal X}$.

We shall now define ${\mathcal E}$, the set of {\em process expressions}, and let $E,F,\ldots$ rage over ${\mathcal E}$.
${\mathcal E}$ is the smallest set
given by the following BNF grammar:
$$\begin{array}{l}
	\begin{array}{ccc|c|c|c|c|c|c}
		E&::=&\alpha.\biguplus_{j\in J}\langle p_{j}\rangle E_{j}&\sum_{i\in I}E_{i}&E\mathbin{|}E &E[f]& E\backslash L& X &A\end{array}
\end{array}
$$

	where $I$ and $J$ are finite indexing sets, $p_j\in(0,1]$ for $j\in J$ and $\sum_{j\in J}p_{j}=1$, $\alpha\in\mathit{Act}$, $L\subseteq\mathcal{L}$,
	and $f$ is a relabelling function.
	Here $\alpha.\biguplus_{j\in J}\langle p_{j}\rangle E_{j}$ is a prefixed expression which can be understood as follows: it can
	perform an action $\alpha$ and afterwards with probability $p_j$ it will behave as $E_j$; $\Sigma_{i\in I}E_i$
	is the summation which is capable of the actions by every $E_i$ for $i\in I$, in particular
	we define $\mathbf{0}=\sum_{i\in \emptyset}\alpha_{i}.E_{i}$ when $I=\emptyset$, i.e.
	$\mathbf{0}$ is the null expression, which is not capable of any action, and moreover when $I=\{1,2\}$
	we will often write $E_1+E_2$ for $\Sigma_{i\in I}E_i$;
	$E_{1}\mathbin{|}E_{2}$ is parallel composition, in which $E_{1}$ and $E_{2}$ can proceed concurrently;
	$E[f]$ is a relabelling expression, which adjusts the labels of actions according to $f$;
	$E\backslash L$ is a restriction expression, which forbids actions with labels in $L \cup \overline L$;
	$X$ is a process variable and $A$ is a constant process. We will use $\equiv$ for syntactic equivalence of expressions.
	We call expressions that without occurrence of any variable processes,
	and we write ${\mathcal P}$ for the set of processes, and use $P,Q,O,U,V\ldots$ to rage over ${\mathcal P}$.
	We assume that every constant process $A$ is associated with a defining equation of the form
	$A\stackrel{{\rm def}}{=}P$, where $A$ can again occur in process $P$ to facilitate recursion.
	
	The operational semantics of processes is given by a transition relation
	$$\mathord{\arrow{}}\subseteq\mathcal{P}\times \mathit{Act}\times{\sf Disc}(\mathcal P),$$
	where
	${\sf Disc}(\mathcal P)$ is the set of discrete distribution functions over ${\mathcal P}$,
	and we will use $\Delta,\Theta,\Omega,\ldots$ to range over ${\sf Disc}({\mathcal P})$.
	${\sf Disc}({\mathcal P})$ can be formally defined as follows.
	\begin{definition} ${\sf Disc}({\mathcal P})$ is the smallest set generated by the
		following rules:
		\begin{enumerate}
			\item For every $P\in{\mathcal P}$, $\delta(P)\in{\sf Disc}({\mathcal P})$
			is the {\em point distribution} which assigns probability $1$
			to $P$ and $0$ to all other $P'\in{\mathcal P}$.
			\item If $\Delta_1,\ldots,\Delta_n\in{\sf Disc}({\mathcal P})$ and $p_1,\ldots,p_n\in(0,1]$ with
			$\Sigma_{i=1}^np_i=1$, let $\Sigma_{i=1}^np_i\Delta_i$ be the function
			s.t. $(\Sigma_{i=1}^np_i\Delta_i)(P)=\Sigma_{i=1}^np_i\Delta_i(P)$
			for every $P\in{\mathcal P}$, then $\Sigma_{i=1}^np_i\Delta_i\in{\sf Disc}({\mathcal P})$.
			We also write $\Sigma_{i=1}^2p_i\Delta_i$ as $p_1\Delta_1+p_2\Delta_2$.
		\end{enumerate}
	\end{definition}
	
	To work with the operational semantics of probabilistic processes, we
	only need the following simple facts about ${\sf Disc}({\mathcal P})$ of which the direct proofs are omitted.
	\begin{proposition} For all $\Delta\in{\sf Disc}(\mathcal P)$ the following hold:
		\begin{enumerate}
			\item $\Delta$ is a distribution function over $\mathcal{P}$, i.e.
			$\Delta(P)\in[0,1]$ for
			$P\in{\mathcal P}$, and $\Sigma_{P\in{\mathcal P}}\Delta(P)=1$.
			\item the set $\setof{P \in \mathcal{P}}{\Delta(P) > 0}$, called the {\em support of} $\Delta$ and
			denoted $\lceil\Delta\rceil$, is a finite set.
		\end{enumerate}
	\end{proposition}
	
	We can now define the operational semantics for the processes.
	\begin{definition}\label{Opsemantics}
		The transition relation
		$$\mathord{\arrow{}}\subseteq\mathcal{P}\times \mathit{Act}\times
		{\sf Disc}(\mathcal{P})$$
		is the least relation generated by the axiom and rules bellow, following the tradition we will write
		$P\arrow{\alpha}\Delta$ instead of $(P,\alpha,\Delta)\in\arrow{}$:
$$
\begin{array}{l}
	\begin{array}{llll}
		\textbf{Act} & \alpha.\biguplus_{j\in J}\langle p_{j}\rangle P_{j}\arrow{\alpha}\Sigma_{j\in J}p_j\delta(P_j)\mbox{}
		\ \ \ \ \ \ \ \ \ &
		\textbf{Sum}_j&\frac{\displaystyle P_j\arrow{\alpha}\Delta_j}{
			\displaystyle \Sigma_{i\in I}P_i\arrow{\alpha}\Delta_j}{\ \ (j\in I)}
	\end{array}\\
	
	\begin{array}{llll}
		\textbf{Com}_1&\frac{\displaystyle P\arrow{\alpha}\Delta}{\displaystyle P|Q\arrow{\alpha}
			\Sigma_{O\in\lceil\Delta\rceil}\Delta(O)\delta({O|Q})}{}\ \ \ \ \ \ \
		&\mbox{}\ \
		\textbf{Com}_2&\frac{\displaystyle Q\arrow{\alpha}\Delta}{\displaystyle P|Q\arrow{\alpha}
			\Sigma_{O\in\lceil\Delta\rceil}\Delta(O)\delta({P|O})}{}
	\end{array}\\
 
	\begin{array}{rl}\textbf{Com}_3&\frac{\displaystyle P\arrow{l}\Delta\ \
			Q\arrow{\bar l}\Theta}{\displaystyle P|Q\arrow{\tau}
			\Sigma_{(U,V)\in\lceil\Delta\rceil\times\lceil\Theta\rceil}(\Delta(U)\cdot\Theta(V))\delta({U|V})}{}
	\end{array}\\
	\begin{array}{llll}\textbf{Rel} & \frac{\displaystyle P\arrow{\alpha}\Delta}{\displaystyle P[f]\arrow{f(\alpha)}
			\Sigma_{O\in\lceil\Delta\rceil}\Delta(O)\delta({O[f]})}{}\ \ \ \ \ \ 
		&\mbox{}\;\;\;\; \
		\textbf{Con} & \frac{\displaystyle P\arrow{\alpha}\Delta}{\displaystyle A\arrow{\alpha}\Delta}
		{\ \ (A\stackrel{{\rm def}}{=}P)}
	\end{array}\\
	\begin{array}{rl}
		\textbf{Res}& \frac{\displaystyle P\arrow{\alpha}\Delta}{\displaystyle P\backslash L\arrow{\alpha}
			\Sigma_{O\in\lceil\Delta\rceil}\Delta(O)\delta({O\backslash L})}
		{\ \ (\alpha,\overline\alpha\notin L)}\end{array}
\end{array}
$$
	\end{definition}
	
	It is not difficult to see that the relation $\arrow{}$ is well-defined. For example, let us
	look at the rule ${\bf Com}_1$, since $\Delta\in{\sf Disc}(\mathcal{P})$, thus
	$\Sigma_{O\in\lceil\Delta\rceil}\Delta(O)=1$, and $\Sigma_{O\in\lceil\Delta\rceil}\Delta(O)\delta({O|Q})
	\in{\sf Disc}(\mathcal{P})$. \\
	Similarly, for the more complex
	${\bf Com}_3$, since $\Delta,\Theta\in{\sf Disc}(\mathcal{P})$ we have
	$$\begin{array}{l}
		\Sigma_{(U,V)\in\lceil\Delta\rceil\times\lceil\Theta\rceil}\Delta(U)\cdot\Delta(V)=
		(\Sigma_{U\in\lceil\Delta\rceil}\Delta(U))\cdot
		(\Sigma_{V\in\lceil\Theta\rceil}\Theta(V))=1,
	\end{array}$$
	thus $\Sigma_{(U,V)\in\lceil\Delta\rceil\times\lceil\Theta\rceil}(\Delta(U)\cdot\Delta(V))\delta({U|V})
	\in{\sf Disc}(\mathcal{P})$.

\section{Branching Bisimilarity}\label{sec4}

In this section we first prepare necessary preliminaries by presenting strong bisimulation for probabilistic processes.
Before defining what is a bisimulation, we need to lift a binary relation on $\mathcal{P}$ to a binary relation
on ${\sf Disc}({\mathcal P})$. Extensive research has already been conducted in this area, as seen in references ~\cite{LarsenS89,DengGHM09,Deng18,Jonsson0L01,CastiglioniT20}, among others. In this paper, we directly adopt the definition of relation lifting from reference~\cite{DengGHM09,Deng18}.

\begin{definition}\label{lifting-relation} Let $R\subseteq{\mathcal P}\times{\mathcal P}$ be a binary relation on
the set of processes. The lifting of $R$, written $R^\dagger$, is the smallest binary relation on
${\sf Disc}({\mathcal P})$ that satisfies:
\begin{enumerate}
\item $PR\,Q$ implies $\delta(P)R^\dagger\delta(Q)$;
\item $\Delta_iR^\dagger\Theta_i$ for $i\in I$ implies
$(\Sigma_{i\in I}p_i\Delta_i)R^\dagger(\Sigma_{i\in I}p_i\Theta_i)$,
where $I$ is an index set s.t.
$\Sigma_{i\in I}p_i=1$.
\end{enumerate}
\end{definition}

The following proposition is often used as an alternative definition of lifting. It is used in many proofs~\cite{BaierEM00,Hennessy12,DengH13}.

\begin{proposition}\label{Prop-DefofLifting}
Let $\Delta,\Theta\in{\sf Disc}(\mathcal{P})$, and $R\subseteq{\mathcal P}\times{\mathcal P}$.
Then the following are equivalent conditions:
\begin{enumerate}
\item $\Delta\,R^\dagger\,\Theta$.
\item there exist $P_1,\ldots,P_n,Q_1,\ldots,Q_n\in\mathcal{P}$ and
$p_1,\ldots,p_n\in(0,1]$ such that $\Sigma_{i=1}^np_i=1$ and $P_i\,R\,Q_i$ for $i=1,\ldots,n$ and
$\Delta,\Theta$ can be decomposed as follows:
\begin{enumerate}
\item $\Delta=\Sigma_{i=1}^np_i\delta(P_i)$;
\item $\Theta=\Sigma_{i=1}^np_i\delta(Q_i)$.
\end{enumerate}
\item there exists a weight function $w:\mathcal{P}\times\mathcal{P}\rightarrow [0,1]$ such that
whenever $w(P,Q)>0$ then $PRQ$, and moreover
\begin{enumerate}

\item for all $P\in\mathcal{P}$ it holds that $\Sigma_{Q\in\mathcal{P}}w(P,Q)=\Delta(P)$;
\item for all $Q\in\mathcal{P}$ it holds that $\Sigma_{P\in\mathcal{P}}w(P,Q)=\Delta(Q)$.
\end{enumerate}
\end{enumerate}
\end{proposition}

The proof of this proposition and that of the following proposition
about lifting can be found in \cite{deng15}.


\begin{proposition}\label{Prolift}
	Let $R_1,R_2\subseteq\mathcal{P}\times\mathcal{P}.$
\begin{enumerate}
\item If $R_1\subseteq R_2$ then $R_1^\dagger\subseteq R_2^\dagger$.
\item $(R_1\cdot R_2)^\dagger=R_1^\dagger\cdot R_2^\dagger$.
\end{enumerate}
\end{proposition}

The following  definition of strong bisimulation is a common
version of strong bisimulation for probabilistic processes, for example~\cite{deng15,Jonsson01}.
This version of bisimulation is different from a weaker version of {\em probabilistic bisimulation} in
that it matches a transition step by an individual transition instead of a {\em combined transition}
which is a combination of several transitions.

\begin{definition}\label{Strongbisim}
A symmetric  binary relation $R\subseteq\mathcal{P}\times\mathcal{P}$ is
a {\em strong bisimulation} if $P{R}\,Q$ implies that whenever
$P\arrow{\alpha}\Delta$
for some $\alpha\in{Act}$ and $\Delta\in{\sf Disc}(\mathcal{P})$, then there exists
$\Theta\in{\sf Disc}(\mathcal P)$ such that $Q\arrow{\alpha}\Theta$ and $\Delta\, R^\dagger\Theta$.

We say that two processes $P,Q\in\mathcal P$ are strong bisimilar, written $P\sim Q$,
if there is a strong bisimulation $R$ such that $PR\,Q$.
\end{definition}

The relation $\sim$ is called strong bisimilarity, and it is the union of all strong bisimulation relations. 
It is standard to  prove that strong bisimilarity is an equivalence relation, and it is also the largest strong bisimulation.

Strong bisimilarity  is
normally considered of only theoretical interest. It is too restrictive to be useful in actual
verification, as it lacks the ability to abstract from internal computations. 
As we have said in the introduction, here we will propose a new notion of branching bisimulation 
for probabilistic processes, which generalizing the original notion for nondeterministic processes~\cite{GlabbeekW96}.
First we need generalize the weak transition of the nondeterministic processes in the following definition.

\begin{definition}
The weak transition
$\dobarrow{}\subseteq{\mathcal P}\times 2^{\mathcal P}$ is defined
by the following rules:
$$
	P \dobarrow{}\{P\} \ \ \ \ \ \ \ \ 
	{P \dobarrow{} \lceil\Delta\rceil}\ (P \arrow{\tau} \Delta)\ \ \ \ \ \ 
	\frac{P\dobarrow{}\{P_1,\ldots,P_n\}, P_1\dobarrow{}S_1,\ldots,P_n\dobarrow{}S_n}{P\dobarrow{}\bigcup_{i=1}^nS_i}
	$$

For ease of notation, we also introduce a transition relation
$\arrow{(\alpha)}\subseteq\mathcal{P}\times{\sf Disc}(\mathcal{P})$ for each $\alpha\in{Act}$, defined by the following two rules:
$$
P\arrow{(\tau)}\delta(P)\ \ \ \ \ \ \ \ \ \ \ P\arrow{(\alpha)}\Delta\  (P\arrow{\alpha}\Delta)
$$
\end{definition}

Here, from a starting state a weak transition end with a set of states.
Basically a weak transition consists of a sequence of  $\tau$ transitions while removing the probability information, or
it could be understood as a sequence of $\tau$ transitions that end in a set of states with probability one.
The non-probabilistic version of transition $\arrow{(\alpha)}$ first appeared in \cite{Basten96}. Next, we prove a lemma about the weak transition relation.






We now define our branching bisimulation for probabilistic processes. The definition is an adaptation of
{\em semi-branching bisimulation} in \cite{Basten96}, which is slightly more general than the original
definition of branching bisimulation of \cite{GlabbeekW96}. Although the two definitions arrive at the
same final equivalence relation, but only the former is closed under relation composition, and this property
is very helpful in proving transitivity of the final relation.

\begin{definition}\label{BranchingBisimulation}
%

Let $R\subseteq\mathcal{P}\times\mathcal{P}$, for two sets $S,T$ of processes we write 
$S\sqsupseteq_RT$ if whenever $Q\in T$ there is $P\in S$ such that $(P,Q)\in R$
  
For $R\subseteq\mathcal{P}\times\mathcal{P}$, define the relation $\mathcal{B}(R)\subseteq\mathcal{P}\times\mathcal{P}$
such that for $P,Q\in\mathcal{P}$, $(P,Q)\in\mathcal{B}(R)$ if
the following hold:
whenever
$P\arrow{\alpha}\Delta$ for some $\alpha\in{Act}$ and $\Delta\in{\sf Disc}(\mathcal P)$,
then $Q\dobarrow{}T$ for some
$T\subseteq\mathcal{P}$ such that $\{P\}\sqsupseteq_R T$, and moreover 
for every $Q'\in T$ there exists $\Theta'\in{\sf Disc}(\mathcal P)$ such that
$Q'\arrow{(\alpha)}\Theta'$ and
$\Delta R^\dagger\Theta'$.
%

A symmetric  binary relation $R\subseteq\mathcal{P}\times\mathcal{P}$ is a
{\em branching bisimulation} if $R\subseteq\mathcal{B}(R)$.

We say that two processes $P,Q$ are branching bisimilar, written $P\approx_{b}Q$,
if there is a branching bisimulation $R$ such that $PR\,Q$.
\end{definition}

Here we briefly discuss the intuition behind this definition. If two processes $P$ and $Q$ are to be branching bisimilar and
$P$ can do an $\alpha$ transition to a distribution $\Delta$, then $Q$ is required to match this by a sequence of $\tau$ transitions
which with probability one end in a set of states $T$, such that every single state $Q'$ in $T$ in fact does not change equivalence 
class, and at the same time can produce a single transition to match $\alpha$. This way of matching the $\alpha$ transition 
follows almost exactly the branching structure of the two compared parts, which is more in line with the spirit of the
original branching bisimulation.

In Definition~\ref{BranchingBisimulation}, inspired by literature\cite{ErkensRL20,SunJLZ23}, the relation $\mathcal{B}(R)$ is first introduced as a branching simulation relation. The relation $\approx_{b}$ is called branching bisimilarity, and it is the union of all probabilistic branching bisimulation relations.
Next, the relationship between probabilistic strong bisimulation and probabilistic branching bisimulation is elaborated through Theorem~\ref{Th-StrongVSbranching}.

\begin{theorem}\label{Th-StrongVSbranching} Every strong bisimulation is a branching bisimulation. Strong bisimilarity is
stronger than branching bisimilarity, i.e. $\sim\ \subseteq\ \approx_b$.
\end{theorem}

An important task of this section is to establish that the probabilistic branching bisimilarity $\approx_b$
is an equivalence relation. We begin by stating several preparatory propositions and lemmas, before culminating in the final result.

\begin{proposition}\label{monotonic} $\mathcal{B}$ is monotonic, i.e. for $R_1,R_2\subseteq\mathcal{P}\times\mathcal{P}$,
if $R_1\subseteq R_2$ then $\mathcal{B}(R_1)\subseteq\mathcal{B}(R_2)$.
\end{proposition}

The following useful property about branching bisimulation
is almost a restatement of the definition, and is useful in the following development.
\begin{lemma}\label{Lemma-BranchingBisimulation}
Let $R\subseteq\mathcal{P}\times\mathcal{P}$ be a
branching bisimulation, then $PR\,Q$ implies the following:\\
whenever
$P\arrow{(\alpha)}\Delta$ for some $\alpha\in{Act}$ and $\Delta\in{\sf Disc}(\mathcal P)$,
then $Q\dobarrow{}T$ for some
$T\subseteq\mathcal{P}$ such that 
for every $Q'\in T$ there exists $\Theta'\in{\sf Disc}(\mathcal P)$ with
$Q'\arrow{(\alpha)}\Theta'$ and
$\Delta R^\dagger\Theta'$.
\end{lemma}



\begin{lemma}\label{weaktransition}
%
Let $R\subseteq\mathcal{P}\times\mathcal{P}$ such that  $R\subseteq\mathcal{B}(R)$. If
$PR\,Q$ and $P\dobarrow{}S$ for some $S\subseteq\mathcal{P}$,
then $Q\dobarrow{}T$ for some $T\subseteq\mathcal {P}$ such that 
$S\sqsupseteq_RT$.
\end{lemma}
\begin{proof} The proof is carried out by induction on the transition
$P\dobarrow{}S$.
If $S$ is $\{P\}$, then take $\{Q\}$ as $T$, in this case obviously we have $Q\dobarrow{}T$ and $S\sqsupseteq_RT$.

If $P\dobarrow{}S$ is because $P\arrow{\tau}\Delta$ where $S=\lceil\Delta\rceil$, since $R\subseteq\mathcal{B}(R)$, thus 
$(P,Q)\in\mathcal{B}(R)$,
then $Q\dobarrow{}T'$ for some
$T'\subseteq\mathcal{P}$ such that  $\{P\}\sqsupseteq_RT'$,  and moreover 
for every $Q'\in T'$ it holds that 
$Q'\arrow{(\tau)}\Theta'$ (which clearly implies that $Q'\dobarrow{}\lceil\Theta'\rceil$) for some $\Theta'\in{\sf Disc}(\mathcal P)$
with $\Delta R^\dagger\Theta'$. 
Thus, let $T'=\{Q_1,\ldots,Q_n\}$
then there exist $\Theta_1,\ldots,\Theta_n$ such that $Q_i\arrow{(\tau)}\Theta_i$ and $\Delta R^\dagger\Theta_i$ (which implies 
$Q_i\dobarrow{}\lceil\Theta_i\rceil$ and
$\lceil\Delta\rceil\sqsupseteq_R
\lceil\Theta_i\rceil$ and $\lceil\Delta\rceil\sqsupseteq_R \bigcup^n_{i=1}\lceil\Theta_i\rceil$
)
for $i=1,\ldots,n$.  To sum up,  $Q\dobarrow{}\{Q_1,\ldots,Q_n\}$, and $Q_1\dobarrow{}\lceil\Theta_1\rceil, \ldots,
Q_n\dobarrow{}\lceil\Theta_n\rceil$, so $Q\dobarrow{}\bigcup_{i=1}^n\lceil\Theta_i\rceil$.

If $P\dobarrow{}S$ is because $P\dobarrow{}\{P_1,\ldots,P_n\}, P_1\dobarrow{}S_1,\ldots,P_n\dobarrow{}S_n$ where
$S=\bigcup_{i=1}^nS_i$, then by the induction hypothesis
there is $T'\subseteq\mathcal P$ such that $Q\dobarrow{}T'$ and
$\{P_1,\ldots,P_n\}\sqsupseteq_RT'$. Let $T'=\{Q_1,\ldots,Q_m\}$, then for each $Q_j\in T'$,  there is $P_i$ with $P_i R Q_j$, 
and for $P_i\dobarrow{}S_i$, again by the induction hypothesis
there is $T_j\subseteq\mathcal{P}$ such that $Q_j\dobarrow{}T_j$ and $S_i\sqsupseteq_R T_j$.
Thus $Q\dobarrow{}\bigcup_{j=1}^mT_j$. Let $T=\bigcup_{j=1}^m$, and it is easy to see that for every $Q'\in T$
there is $P'\in S$ such that $P' R Q'$, that is $S\sqsupseteq_RT$.
\end{proof}

\begin{lemma}\label{Lemma-transitivity} Let $R_1,R_2\subseteq\mathcal{P}\times\mathcal{P}$. If $R_1\subseteq\mathcal{B}(R_1)$ and
$R_2\subseteq\mathcal{B}(R_2)$, then 
$R_1\cdot R_2\subseteq\mathcal{B}(R_1\cdot R_2)$.
\end{lemma}

\begin{proposition}\label{PropertyofBranchingBisimulation}
Assume that each $R_i$ ($i\in I$) is a branching bisimulation, then
the following relations
are all branching bisimulations:
\begin{enumerate}
\item ${\mathcal Id}_{\mathcal{P}}$,
\item $R_1\cdot R_2\cup R_2\cdot R_1$,
\item $\bigcup_{i\in I}R_i$.
\end{enumerate}
\end{proposition}
\begin{proof}
The proof of 1):
Obviously, ${\mathcal Id}_{\mathcal{P}}$ is a strong bisimulation, and by Theorem
\ref{Th-StrongVSbranching}
every strong bisimulation is a branching bisimulation, hence
${\mathcal Id}_{\mathcal{P}}$ is a branching bisimulation.

The proof of 2): Since $R_1, R_2$ are branching bisimulations, they are symmetric, thus
$R_1\cdot R_2\cup R_2\cdot R_1$ is symmetric, also $R_1\subseteq\mathcal{B}(R_1)$ and $R_2\subseteq
\mathcal{B}(R_2)$, thus by Lemma \ref{Lemma-transitivity} $R_1\cdot R_2\subseteq\mathcal{B}(R_1\cdot R_2)$
and $R_2\cdot R_1\subseteq\mathcal{B}(R_2\cdot R_1)$, thus
$R_1\cdot R_2\cup R_2\cdot R_1\subseteq\mathcal{B}(R_1\cdot R_2\cup R_2\cdot R_1)$ easily follows from 
Proposition \ref{monotonic}, and $R_1\cdot R_2\cup R_2\cdot R_1$ is a branching bisimulation.

We omit the proof of 3), which is straightforward.
\end{proof}

Finally, it is shown via Theorem~\ref{Branchingisequai} that $\approx_b$ is an equivalence relation.

\begin{theorem}\label{Branchingisequai} $\approx_b$ is the largest branching bisimulation, and it is an equivalence relation.
\end{theorem}
\begin{proof} Notice that
$$\approx_b=\setof{R}{R\mbox{ is a branching bisimulation}},$$
then according to
3) of Proposition \ref{PropertyofBranchingBisimulation}, $\approx_b$ is a branching bisimulation.
If $R$ is a branching bisimulation, then clearly $R\subseteq\approx_b$. Hence
$\approx_b$ is the largest branching bisimulation.

By 1) of Proposition \ref{PropertyofBranchingBisimulation}
${\mathcal Id}_{\mathcal{P}}=\setof{(P,P)}{P\in\mathcal{P}}$ is a branching bisimulation,
thus $P\approx_bP$ for every $P\in\mathcal{P}$, that is $\approx_b$ is reflexive.
As a branching bisimulation $\approx_b$ is clearly symmetric. To see that $\approx_b$ is transitive,
suppose $P\approx_bO, O\approx_bQ$, then there are branching bisimulations $R_1,R_2$ such that
$PR_1O,OR_2Q$, thus $PR_1\cdot R_2Q$ and $P(R_1\cdot R_2\cup R_2\cdot R_1)Q$. By Proposition \ref{PropertyofBranchingBisimulation}
$R_1\cdot R_2\cup R_2\cdot R_1$ is a branching
bisimulation, thus $P\approx_bQ$, which shows that $\approx_b$ is transitive. Hence $\approx_b$ is an equivalence relation.
\end{proof}


\section{Branching Equality}\label{sec5}
In this section, we first prove that branching bisimilarity is preserved by Composition, Restriction, and Relabeling. As Milner \cite{milner89} demonstrated, weak bisimilarity is not fully substitutive and is not preserved by Summation due to the preemptive nature of $\tau$-actions. For example, although \( P \approx \tau.P \), we have \( P + Q \not\approx \tau.P + Q \). The same applies to our branching bisimilarity. To remedy this,  Milner \cite{milner89} introduced the rootedness condition.Therefore, in this section we draw on this idea and introduce a notion of equality~---~branching equality,  denoted by $=_b$, that is fully substitutive. We prove that $=_b$ is a congruence for all the operators, including recursion.

\begin{lemma}\label{decom} Let $P,Q\in\mathcal{P}$.
If $P\dobarrow{}S$ and $Q\dobarrow{}T$ for $S,T\subseteq\mathcal{P}$, then
\begin{enumerate}
\item For all $O\in\mathcal{P}$,
$P|O\dobarrow{}\setof{P'|O}{P'\in S}$ and
$O|P\dobarrow{}\setof{O|P'}{P'\in S}$;
\item
$P|Q\dobarrow{}\setof{P'|Q'}{P'\in S, Q'\in T}$;
\item For all $L\subseteq \mathcal{L}$,
$P\backslash L\dobarrow{}\setof{P'\backslash L}{P'\in S}$;
\item For all relabeling  function $f$,
$P[f]\dobarrow{}\setof{P'[f]}{P'\in S}$.
\end{enumerate}
\end{lemma}
\begin{proof} We prove item 1 by induction on the transition $P\dobarrow{}S$.
If $S=\{P\}$, then it is obvious that $P|O\dobarrow{}\{P|O\}$, the claim holds in this case.
If $P\dobarrow{}S$ is due to $P\arrow{\tau}\Delta$ where $S=\lceil\Delta\rceil$,
then by ${\bf com}_1$ of Definition \ref{Opsemantics}
we have $P|O\arrow{\tau}\Sigma_{P'\in\lceil\Delta\rceil}\Delta(P')\delta(P'|O)$.
Therefore,
$P|O\dobarrow{}\lceil\Sigma_{P'\in\lceil\Delta\rceil}\Delta(P')\delta(P'|O)\rceil$, 
and it is easy to see that $\lceil\Sigma_{P'\in\lceil\Delta\rceil}\Delta(P')\delta(P'|O)\rceil=
\setof{P'|O}{P'\in\lceil\Delta\rceil}$, so the claim holds in this case.
Now suppose $P\dobarrow{}S$ is due to $P\dobarrow{}\{P_1,\ldots,P_n\}$,
$P_1\dobarrow{}S_1,\ldots, P_n\dobarrow{}S_n$ where $S=\bigcup_{i=1}^nS_i$,
then by the induction hypothesis $P|O\dobarrow{}\{P_1|O,\ldots,P_n|O\}$,
$P_1|O\dobarrow{}\setof{P'|O}{P'\in S_1},\ldots, P_n|O\dobarrow{}\setof{P'|O}{P'\in S_n}$,
then in this case $P|O\dobarrow{}\bigcup_{i=1}^n\setof{P'|O}{P'\in S_i}=\setof{P'|O}{P'\in\bigcup_{i=1}^nS_i}$,
so the claim holds in this case.

The proof of the symmetric case is similar, and thus omitted.

For item 2, according to item 1, we have $P|Q\dobarrow{}\setof{P'|Q}{P'\in S}$.
and moreover, for each $P'\in S$, we have
$P'|Q\dobarrow{}\setof{P'|Q'}{Q'\in T}$.
Thus
$P|Q\dobarrow{}\bigcup_{P'\in S}\setof{P'|Q'}{Q'\in T}=\setof{P'|Q'}{P'\in T,Q'\in S}$, so the claim holds in this case.

The proofs for items 3 and 4 follow a similar structure to that of item 1 and thus omitted.
\end{proof}

\begin{proposition}\label{branch-congruence}
	Let $P,Q\in\mathcal{P}$. If $P\approx_b Q$ then
\begin{enumerate}
	\item for any $O\in\mathcal{P}$, $P\mathbin{|}O\approx_b Q\mathbin{|}O$;
	\item for any $L\subseteq \mathcal{L}$, $P\backslash L\approx_b Q\backslash L$;
	\item for any relabeling  function $f$, $P[f]\approx_b Q[f]$.
\end{enumerate}
\end{proposition}
\begin{proof}
\textbf{The proof of 1:}
Let $O\in\mathcal{P}$. To prove $P|O\approx_b Q|O$, it suffices to show
that the following binary relation $R$ is a branching bisimulation:
$$
R=\{(P|O, Q|O) \mid O,P,Q\in\mathcal{P} \text{ and } P \approx_b Q\}.
$$
First, since $\approx_b$ is symmetric, $R$ is also symmetric.
Suppose $(P|O,Q|O)\in R$ and $P|O\arrow{\alpha}\Delta$.
To complete the proof, we need to show that there exists some $T\subseteq
\mathcal{P}$ such that $Q|O\dobarrow{}T$ with
$\{P|O\}\sqsupseteq _RT$, and moreover for every $U\in T$, it holds that there exists $\Theta'\in{\sf Disc}(\mathcal P)$ such that
$U\arrow{(\alpha)}\Theta'$ and $\Delta R^\dagger\Theta'$.
To proceed with the proof, we analyze the possible transitions based on the operational semantics defined in Definition \ref{Opsemantics}. Specifically, we identify three cases:\\
\textbf{Case 1:} In this case, $P\arrow{\alpha}\Delta_1$ and
$\Delta= \Sigma_{P'\in\lceil\Delta_1\rceil}\Delta_1(P')\delta({P'|O})$.
Since $P\approx_{b} Q$ and $\approx_b$ is a branching bisimulation,  there exists $T_1\subseteq\mathcal{P}$ such that
$Q\dobarrow{}T_1$ with $\{P\}\sqsupseteq_{\approx_{b}}T_1$, and moreover for every $Q'\in T_1$ there is $\Theta'_1\in{\sf Disc}(\mathcal P)$ such that
$Q'\arrow{(\alpha)}\Theta'_1$ and $\Delta_1({\approx_b})^\dagger\Theta'_1$.
By Lemma~\ref{decom}, we have $Q|O \dobarrow{}T$ where $T=\setof{Q'|O}{Q'\in T_1}$.
Since $P\approx_{b} Q'$ for every $(Q'|O)\in T$, we have $\{P|O\}\sqsupseteq_R T$ as required.
Moreover, for the $\Theta'_1$ with $Q'\arrow{(\alpha)}\Theta'_1$ and $\Delta_1(\approx_{b})^\dagger\Theta'_1$,
it holds that $Q'|O\arrow{(\alpha)}\Sigma_{Q''\in\lceil\Theta_1'\rceil}\Theta_1'(Q'')\delta({Q''|O})$. Therefore 
there exists $\Theta'$ such that $(Q'|O)\arrow{(\alpha)}\Theta'$ where
$\Theta'=\Sigma_{Q''\in\lceil\Theta_1'\rceil}\Theta_1'(Q'')\delta({Q''|O})$.
From part 3 of Proposition~\ref{Prop-DefofLifting}, and $\Delta_1(\approx_{b})^\dagger\Theta_1'$, there exists a function $w:\mathcal{P}\times \mathcal{P}\rightarrow [0,1]$ such that:
	\begin{enumerate}
		\item[a)]$P'\approx_{b} Q''$ for all $w(P',Q'')>0$,
		\item[b)]$\Delta_1(P')=\Sigma_{Q''\in\mathcal{P}}w(P',Q'')$ for all $P'\in \mathcal{P}$, and
		\item[c)]$\Theta_1'(Q'')=\Sigma_{P'\in\mathcal{P}}w(P',Q'')$ for all $Q''\in \mathcal{P}$.
	\end{enumerate}
Then we can perform the following calculations:
$$\begin{array}{lclr}
	\Delta&=&\Sigma_{P'\in\lceil\Delta_1\rceil}\Delta_1(P')\delta({P'|O}) \\
	&=&\Sigma_{P'\in\mathcal{P}}\Delta_1(P')\delta({P'|O})& \mbox{$\Delta_1(P')=0$
		for $P'\notin\lceil\Delta_1\rceil$}\\
	&=&
	\Sigma_{P'\in\mathcal{P}}(\Sigma_{Q''\in\mathcal{P}}w(P',Q''))\delta({P'|O})&\mbox{above b)}\\
	&=&
	\Sigma_{P'\in\mathcal{P},Q'\in\mathcal{P}}w(P',Q'')\delta({P'|O})\\
	&=&
	\Sigma_{P'\approx_{b} Q''}w(P',Q'')\delta({P'|O})&\mbox{by a) above, $w(P',Q'')=0$ for $P'\not\approx_{b} Q''$}
\end{array}
$$
Similarly, we can obtain $\Theta'=\Sigma_{P'\approx_{b} Q''}w(P',Q'')\delta({Q''|O})$.
Since $P'\approx_{b} Q''$, we have $(P'|O,Q''|O)\in R$, hence $\Delta R^\dagger \Theta'$, as required.	\\

The proof  of other two cases can be carried out similarly. We omit the proofs of 2 and 3.
\end{proof}

Branching bisimilarity is not a congruence for Summation. For instance, $P\approx_b Q$ does not imply $P+R\approx_b Q+R$. We introduce a new notion of equality, termed branching equality. We formally define this notion below and demonstrate that branching equality is situated between strong congruence and branching bisimilarity.

\begin{definition}\label{branching-equal}
$P$ and $Q$ are branching equal, written $P=_bQ$, 
	if for all $\alpha\in{Act}$
	\begin{enumerate}
		\item Whenever $P\arrow{\alpha}\Delta$ then, for some $\Delta'$, $Q\arrow{\alpha}\Delta'$
		and $\Delta(\approx_b)^\dagger\Delta'$; 
		\item Whenever $Q\arrow{\alpha}\Delta'$ then, for some $\Delta$, $P\arrow{\alpha}\Delta$
		and $\Delta(\approx_b)^\dagger\Delta'$. 
	\end{enumerate}
\end{definition}

\begin{proposition}\label{pro:distinguishing-power}
$\sim$ implies $=_b$, and $=_b$ implies $\approx_b$.
\end{proposition}
\begin{proof} To show that $\sim$ implies $=_b$, assume $P\sim Q$. Suppose $P\arrow{\alpha}\Delta$.
Since $\sim$ is a strong bisimulation, there exists $\Delta'$ such that $Q\arrow{\alpha}\Delta'$
and $\Delta(\sim)^\dagger\Delta'$. By Theorem \ref{Th-StrongVSbranching}, we have $\sim\,\subseteq\,\approx_b$. Therefore,
$\Delta(\approx_b)^\dagger\Delta'$ by Proposition \ref{Prolift}. Similarly, if $Q\arrow{\alpha}\Delta'$,
then there exists $\Delta$ such that $P\arrow{\alpha}\Delta$ and $\Delta(\approx_b)^\dagger\Delta'$. Thus $P=_bQ$, as required.

To show that $=_b$ implies $\approx_b$, it is enough to show that $=_b\cup\approx_b$ is a branching bisimulation, which is immediate.
\end{proof}

To show that branching equality is a congruence relation, we first prove that it is an equivalence relation.

\begin{proposition}\label{branequaisequiva}
	$=_b$ is an equivalence relation.
\end{proposition}
\begin{proof}
To show that $=_b$ is an equivalence relation, we need to prove that it is reflexive, symmetric, and transitive.	\\
\textbf{Reflexivity:} Since $P\sim P$ for all $P\in\mathcal{P}$, and $\sim$ implies $=_b$, it follows that $P=_bP$. Therefore, $=_b$ is reflexive.\\
\textbf{Symmetry:} By definition, it is immediate that $=_b$ is symmetric.\\
\textbf{Transitivity:} To prove that $=_b$ is transitive, assume
$P=_b O$ and $O=_b Q$. Suppose $P\arrow{\alpha}\Delta$.
Since  $P=_b O$, there exists $\Omega$ such that
$O\arrow{\alpha}\Omega \text{ and } \Delta(\approx_b)^\dagger\Omega$.
Also, since $O=_b Q$, there exists $\Theta$ such that
$Q\arrow{\alpha}\Theta \text{ and } \Omega(\approx_b)^\dagger\Theta$.
Given $\Delta(\approx_b)^\dagger\Omega$ and $\Omega(\approx_b)^\dagger\Theta$, we have
$\Delta(\approx_b)^\dagger\cdot(\approx_b)^\dagger\Theta$.
By part 2 of Proposition~\ref{Prolift}, this implies
$\Delta(\approx_b\cdot\approx_b)^\dagger\Theta$. Therefore, $\Delta(\approx_b)^\dagger\Theta$.
Similarly, if $Q\arrow{\alpha}\Theta$, then there exists $\Delta$ such that $P\arrow{\alpha}\Delta$
and $\Delta(\approx_b)^\dagger\Theta$.
Thus, $P=_b Q$, and $=_b$ is transitive.
\end{proof}

We now demonstrate that branching equality is a congruence for all operators.

\begin{theorem}\label{branequaiscongru}  For $P,Q\in\mathcal{P}$, if $P=_bQ$ then the following hold:
\begin{enumerate}
\item $\alpha.(\biguplus_{i=1}^n\langle p_i\rangle P_i)=_b
\alpha.(\biguplus_{i=1}^n\langle p_i\rangle Q_i)$, where $P_j$ is $P$, $Q_j$ is $Q$, and $P_i\equiv Q_i$ for $i\not=j$;
\item $\Sigma_{i=1}^nP_i=_b\Sigma_{i=1}^nQ_i$, where $P_j$ is $P$, $Q_j$ is $Q$, and $P_i\equiv Q_i$ for $i\not=j$;
\item $P|O=_bQ|O, O|P=_bO|Q$,  where $O\in\mathcal{P}$;
\item $P\backslash L=_bQ\backslash L$, where $L\subseteq\mathcal{L}$;
\item $P[f]=_bQ[f]$, where $f$ is a relabeling function.
\end{enumerate}
\end{theorem}
\begin{proof}
Suppose $P=_bQ$. Then we have $P\approx_bQ$ by Proposition \ref{pro:distinguishing-power}.\\
\textbf{The proof of 1:} By Definition~\ref{Opsemantics},
$\alpha.(\biguplus_{i=1}^n\langle p_i\rangle P_i)\arrow{\alpha}\Sigma_{i=1}^{n}p_{i}\delta(P_{i})$
is the only transition from $\alpha.(\biguplus_{i=1}^n\langle p_i\rangle P_i)$. Similarly,
$\alpha.(\biguplus_{i=1}^n\langle p_i\rangle Q_i)\arrow{\alpha}\Sigma_{i=1}^{n}p_{i}\delta(Q_{i})$
is the only transition from $\alpha.(\biguplus_{i=1}^n\langle p_i\rangle Q_i)$.
Given the condition $P_j =_b Q_j$ for $i=j$ and $P_i\equiv Q_i$ for $i\neq j$, we obtain
$P_i\approx_bQ_i$ for $i=1,\ldots,n$. Then, we can apply Proposition \ref{Prop-DefofLifting} to obtain
$\Sigma_{i=1}^{n}p_{i}\delta(P_{i})(\approx_b)^\dagger\Sigma_{i=1}^{n}p_{i}\delta(Q_{i})$.
Therefore, by the definition of $=_b$ and the above result, we conclude that
$\alpha.(\biguplus_{i=1}^n\langle p_i\rangle P_i)=_b
\alpha.(\biguplus_{i=1}^n\langle p_i\rangle Q_i)$.\\

\textbf{The proof of 2:} Suppose $\Sigma_{i=1}^nP_i\arrow{\alpha}\Delta$. By Definition~\ref{Opsemantics}, if the transition is due to $P_j\arrow{\alpha}\Delta$. Since $P_j=_bQ_j$, there exists
$\Theta$ such that $Q_j\arrow{\alpha}\Theta$ and  $\Delta(\approx_b)^\dagger\Theta$.
Therefore, $\Sigma_{i=1}^nQ_i\arrow{\alpha}\Theta$ with $\Delta(\approx_b)^\dagger\Theta$.
On the other hand, if the transition is due to $P_i\arrow{\alpha}\Delta$ for some $i\neq j$. Given $P_i\equiv Q_i$, it follows that $Q_i\arrow{\alpha}\Delta$ and $\Delta(\approx_b)^\dagger\Delta$. Thus, $\Sigma_{i=1}^nQ_i\arrow{\alpha}\Delta$ with $\Delta(\approx_b)^\dagger\Delta$.
Similarly, we can show that if $\Sigma_{i=1}^nQ_i\arrow{\alpha}\Theta$, then
there is $\Delta$ such that $\Sigma_{i=1}^nP_i\arrow{\alpha}\Delta$ and $\Delta(\approx_b)^\dagger\Theta$.
Thus, we conclude that $\Sigma_{i=1}^nP_i=_b\Sigma_{i=1}^nQ_i$.\\

\textbf{The proof of 3:} Since the methods of proof for the two situations are similar, we will only prove the first situation here. Suppose $P|O\arrow{\alpha}\Delta$, we need to show that there exists some $\Theta$ such that $Q|O\arrow{\alpha}\Theta$ and $\Delta(\approx_b)^\dagger\Theta$.
According to Definition~\ref{Opsemantics}, we have three cases.\\

\textbf{Case 1:} In this case $P\arrow{\alpha}\Delta_1$ and
$\Delta= \Sigma_{U\in\lceil\Delta_1\rceil}\Delta_1(U)\delta({U|O})$.
Since $P=_b Q$, by Definition~\ref{branching-equal}, there exists $\Delta_2$
such that
$Q\arrow{\alpha}\Delta_2$ and $\Delta_1(\approx_b)^\dagger \Delta_2$.
Then, by Definition~\ref{Opsemantics}, there is $\Theta=\Sigma_{V\in\lceil\Delta_2\rceil}\Delta_2(V)\delta({V|O})$ such that
$Q|O\arrow{\alpha}\Theta$. To complete the proof, we need to show that $\Delta(\approx_b)^\dagger\Theta$.

From part 3 of Proposition~\ref{Prop-DefofLifting} and $\Delta_1(\approx_b)^\dagger \Delta_2$, there is a function $w:\mathcal{P}\times \mathcal{P}\rightarrow [0,1]$ such that:
\begin{enumerate}
	\item[a)]$U\approx_b V$ for all $w(U,V)>0$,
	\item[b)]$\Delta_1(U)=\Sigma_{V\in\mathcal{P}}w(U,V)$ for all $U\in \mathcal{P}$, and
	\item[c)]$\Delta_2(V)=\Sigma_{U\in\mathcal{P}}w(U,V)$ for all $V\in \mathcal{P}$.
\end{enumerate}
Let
$$\Delta^*=\Sigma_{U\approx_b V}w(U,V)\delta{(U|O)},\Theta^*=\Sigma_{U\approx_b V}w(U,V)\delta{(V|O)}.
$$
Since $U\approx_b V$, we have $(U|O,V|O)\in \approx_b$ by Proposition~\ref{branch-congruence},  and
$$ \Sigma_{U\approx_b V}w(U,V)=\Sigma_{U\in\mathcal{P}, V\in\mathcal{P}}w(U,V)=\Sigma_{U\in\mathcal{P}}\Sigma_{V\in\mathcal{P}}w(U,V)=\Sigma_{U\in\mathcal{P}}\Delta_1(U)=1$$
(the first equality is because $w(U,V)=0$ for $U\not\approx_b V$, and the third is because b) above)
then by parts 1), 2) of Proposition~\ref{Prop-DefofLifting}, we have $\Delta^*(\approx_b)^\dagger\Theta^*$, and
moreover we can perform the following calculations:
$$\begin{array}{lclr}
	\Delta&=&\Sigma_{U\in\lceil\Delta_1\rceil}\Delta_1(U)\delta({U|O})& \mbox{construction of $\Delta$}\\
	&=&\Sigma_{U\in\mathcal{P}}\Delta_1(U)\delta({U|O})& \mbox{$\Delta_1(U)=0$
		for $U\notin\lceil\Delta_1\rceil$}\\
	&=&
	\Sigma_{U\in\mathcal{P}}(\Sigma_{V\in\mathcal{P}}w(U,V))\delta({U|O})&\mbox{above b)}\\
	&=&
	\Sigma_{U\in\mathcal{P},V\in\mathcal{P}}w(U,V)\delta({U|O})\\
	&=&
	\Sigma_{U\approx_b V}w(U,V)\delta({U|O})&\mbox{by a) above, $w(U,V)=0$ for $U\not\approx_b V$}\\
	&=&\Delta^*\\
	\Theta&=&\Sigma_{V\in\lceil\Delta_2\rceil}\Delta_2(V)\delta({V|O})& \mbox{construction of $\Theta$}\\
	&=&\Sigma_{V\in\mathcal{P}}\Delta_2(V)\delta({V|O})& \mbox{$\Delta_2(V)=0$
		for $V\notin\lceil\Delta_2\rceil$}\\
	&=&
	\Sigma_{V\in\mathcal{P}}(\Sigma_{U\in\mathcal{P}}w(U,V))\delta({V|O})&\mbox{above c)}\\
	&=&
	\Sigma_{U\in\mathcal{P},V\in\mathcal{P}}w(U,V)\delta({V|O})\\
	&=&
	\Sigma_{U\approx_b V}w(U,V)\delta({V|O})&\mbox{$w(U,V)=0$ for $U\not\approx_b V$}\\
	&=&\Theta^*
\end{array}
$$
Thus, $\Delta (\approx_b)^\dagger \Theta$, as required. \\

\textbf{Case 2:} The transition of $P|O \arrow{\alpha} \Delta$ is induced by
$O \arrow{\alpha} \Delta_1$ with
$\Delta = \Sigma_{O'\in\lceil\Delta_1\rceil}\Delta_1(O')$ $\delta({P|O'})$,
according to Definition~\ref{Opsemantics} .
Similarly, we obtain $Q|O \arrow{\alpha} \Theta$ with
$\Theta = \Sigma_{O'\in\lceil\Delta_1\rceil}\Delta_1(O')$ $\delta({Q|O'})$.
Given $P =_b Q$, by Proposition~\ref{pro:distinguishing-power} we have $P \approx_b Q$ .
By Proposition~\ref{branch-congruence}, we have $(P|O')\approx_b(Q|O')$.
Consequently, we have $ \delta(P|O') (\approx_b)^\dagger \delta(Q|O')$ by Definition~\ref{lifting-relation}.
Therefore, by Definition~\ref{lifting-relation}, we deduce that $\Sigma_{O'\in\lceil\Delta_1\rceil}\Delta_1(O')\delta({P|O'})(\approx_b)^\dagger\Sigma_{O'\in \lceil\Delta_1\rceil}\Delta_1(O')\delta(Q|O')$ , which simplifies to $\Delta (\approx_b)^\dagger\Theta$.
Thus $R$ is branching equal.\\

\textbf{Case 3:} In this case, $\alpha=\tau$ and the
process undergoes a $\tau$ transition triggered by handshake synchronization.
Specifically, there exist $\Delta',\Omega\in{\sf Disc}({\mathcal P})$ and $l\in\mathcal{A}\cup\overline{\mathcal{A}}$ such that $P\arrow{l}\Delta', O\arrow{\bar l}\Omega$, and
$\Delta=\Sigma_{(U,V)\in\lceil\Delta'\rceil\times\lceil\Omega\rceil}
(\Delta'(U)\cdot\Omega(V))\delta(U|V)$.
Since $P =_b Q$, there exists $\Theta'$ such that $Q \arrow{l} \Theta'$ and
$\Delta'(\approx_b)^\dagger \Theta'$.
By Definition~\ref{Opsemantics}, we have $Q|O\arrow{\tau}\Theta$ where
$\Theta=\Sigma_{(W,V)\in\lceil\Theta'\rceil\times\lceil\Omega\rceil}
(\Theta'(W)\cdot\Omega(V))\delta(W|V)$. To complete the proof, we need to show that $\Delta (\approx_b)^\dagger\Theta$.

From Proposition~\ref{Prop-DefofLifting} and $\Delta'(\approx_b)^\dagger \Theta'$, there exists a function $w:\mathcal{P}\times\mathcal{P}\rightarrow [0,1]$ such that:
\begin{enumerate}
	\item[a)]$U\approx_b W$ for all $w(U,W)>0$,
	\item[b)]$\Delta'(U)=\Sigma_{W\in\mathcal{P}}w(U,W)$ for all $U\in\mathcal{P}$, and
	\item[c)]$\Theta'(W)=\Sigma_{U\in\mathcal{P}}w(U,W)$ for all $W\in \mathcal{P}$.
\end{enumerate}
For each $V\in\lceil\Omega\rceil$, define
$$\Delta_V=\Sigma_{U\approx_b W}w(U,W)\delta(U|V), \ \Theta_V=\Sigma_{U\approx_b W}w(U,W)\delta(W|V).
$$
Since $U\approx_b W$, we have $(U|V,W|V)\in \approx_b$ by Proposition~\ref{branch-congruence}. Moreover,
$$\Sigma_{U\approx_b W}w(U,W)=
\Sigma_{U\in\mathcal{P},W\in\mathcal{P}}w(U,W)=\Sigma_{U\in\mathcal{P}}\Sigma_{W\in\mathcal{P}}w(U,W)=
\Sigma_{U\in\mathcal{P}}\Delta'(U)=1.$$
By parts 1) and 2) of Proposition~\ref{Prop-DefofLifting}, we have
$\Delta_V(\approx_b)^\dagger\Theta_V$. \\
Therefore,
$\Sigma_{V\in\lceil\Omega\rceil}\Omega(V)\Delta_V(\approx_b)^\dagger\Sigma_{V\in\lceil\Omega\rceil}\Omega(V)\Theta_V$. We now perform
the following calculations:
$$\begin{array}{lclr}
	\Delta&=&\Sigma_{(U,V)\in\lceil\Delta'\rceil\times\lceil\Omega\rceil}
	(\Delta'(U)\cdot\Omega(V))\delta(U|V)& \mbox{construction of $\Delta$}\\
	&=&\Sigma_{V\in\lceil\Omega\rceil}\Sigma_{U\in\lceil\Delta'\rceil}
	(\Delta'(U)\cdot\Omega(V))\delta(U|V)\\
	&=&\Sigma_{V\in\lceil\Omega\rceil}\Omega(V)\Sigma_{U\in\lceil\Delta'\rceil}
	\Delta'(U)\delta(U|V)\\
	&=&\Sigma_{V\in\lceil\Omega\rceil}\Omega(V)\Sigma_{U\in\mathcal{P}}
	\Delta'(U)\delta(U|V)&\mbox{$\Delta'(U)=0$
		for $U\notin\lceil\Delta'\rceil$}\\
	&=&\Sigma_{V\in\lceil\Omega\rceil}\Omega(V)\Sigma_{U\in\mathcal{P}}
	\Sigma_{W\in\mathcal{P}}w(U,W)\delta(U|V)&\mbox{above b)}\\
	&=&\Sigma_{V\in\lceil\Omega\rceil}\Omega(V)\Sigma_{U\in\mathcal{P},
		W\in\mathcal{P}}w(U,W)\delta(U|V)\\
	&=&\Sigma_{V\in\lceil\Omega\rceil}\Omega(V)\Sigma_{U\approx_b W}w(U,W)\delta(U|V)&\mbox{$w(U,W)=0$ for $U\not\approx_b W$}\\
	&=&\Sigma_{V\in\lceil\Omega\rceil}\Omega(V)\Delta_V\\

	\Theta&=&\Sigma_{(W,V)\in\lceil\Theta'\rceil\times\lceil\Omega\rceil}
	(\Theta'(W)\cdot\Omega(V))\delta(W|V)& \mbox{construction of $\Theta$}\\
	&=&\Sigma_{V\in\lceil\Omega\rceil}\Sigma_{W\in\lceil\Theta'\rceil}
	(\Theta'(W)\cdot\Omega(V))\delta(W|V)\\
	&=&\Sigma_{V\in\lceil\Omega\rceil}\Omega(V)\Sigma_{W\in\lceil\Theta'\rceil}
	\Theta'(W)\delta(W|V)\\
	&=&\Sigma_{V\in\lceil\Omega\rceil}\Omega(V)\Sigma_{W\in\mathcal{P}}
	\Theta'(W)\delta(W|V)&\mbox{$\Theta'(W)=0$
		for $W\notin\lceil\Theta'\rceil$}\\
	&=&\Sigma_{V\in\lceil\Omega\rceil}\Omega(V)\Sigma_{W\in\mathcal{P}}
	\Sigma_{U\in\mathcal{P}}w(U,W)\delta(W|V)&\mbox{above c)}\\
	&=&\Sigma_{V\in\lceil\Omega\rceil}\Omega(V)\Sigma_{W\in\mathcal{P},
		U\in\mathcal{P}}w(U,W)\delta(W|V)\\
	&=&\Sigma_{V\in\lceil\Omega\rceil}\Omega(V)\Sigma_{U\approx_b W}w(U,W)\delta(W|V)&\mbox{$w(U,W)=0$ for $U\not\approx_b W$}\\
	&=&\Sigma_{V\in\lceil\Omega\rceil}\Omega(V)\Theta_V\\
\end{array}
$$
Thus $\Delta (\approx_b)^\dagger\Theta$. The proof for the symmetric case is analogous, and thus we omit it here.\\

\textbf{The proof of 4,5:} The proofs for cases 4 and 5 are analogous, and thus we omit them here.
\end{proof}

Next, we prove that $=_b$ is a congruence for recursion. To this end, we introduce a relation known as branching bisimulation up to, which we will show to be a branching bisimulation. Utilizing this technique allows us to simplify the proof process.

\begin{definition}\label{bran-upto}
	A symmetric  binary relation $R\subseteq\mathcal{P}\times\mathcal{P}$ is a
	{\em branching bisimulation up to $ \approx_b $} if $PR\,Q$ implies the following:
	whenever
	$P\arrow{\alpha}\Delta$ for some $\alpha\in{Act}$ and $\Delta\in{\sf Disc}(\mathcal P)$,
	then 
	$Q\arrow{\alpha}\Theta$ for some
	$\Theta\in{\sf Disc}(\mathcal{P})$ such that
	$\Delta (R\cdot\approx_b)^\dagger\Theta$. 
\end{definition}

\begin{lemma}\label{Lemma-upto}
Let $R$ be a branching bisimulation up to $\approx_b$, then
$R\cdot\approx_b\subseteq\mathcal{B}(R\cdot\approx_b)$.
\end{lemma}
\begin{proof}
Assume $P(R\cdot\approx_b)Q$ and $P\arrow{\alpha}\Delta$. To complete this proof, we need to show that there exists some $T\subseteq\mathcal{P}$ 
such that $Q\dobarrow{}T$ and $\{P\}\sqsupseteq_{R\cdot\approx_b} T$, and  moreover for every $Q' \in T$
there exists $\Theta'$ such that $Q'\arrow{(\alpha)}\Theta'$ and $\Delta (R\cdot\approx_b)^\dagger \Theta'$. 

For $P(R\cdot\approx_b)Q$, there exists $P_1$ such that $PRP_1$ and $P_1\approx_b Q$. 
Since $R$ is a branching bisimulation up to $\approx_b$, $PRP_1$ and $P\arrow{\alpha}\Delta$ implies that
there exists $\Theta_1$ such that $P_1\arrow{\alpha}\Theta_1$ and $\Delta (R\cdot\approx_b)^\dagger \Theta_1$.
Since $P_1\approx_b Q$, there exists $T\subseteq\mathcal{P}$ such that $Q\dobarrow{}T$ and $\{P_1\}\sqsupseteq_{\approx_b}T$, 
and moreover for every $Q' \in T$ we have $P_1 \approx_b Q'$.
Since $PRP_1$ and $P_1 \approx_b Q'$, it follows that $P(R\cdot\approx_b)Q'$. Moreover 
for every $Q' \in T$, there exists $\Theta'$ such that $Q'\arrow{(\alpha)}\Theta'$, and $\Theta_1 (\approx_b)^\dagger \Theta'$.
Since $\Delta (R\cdot\approx_b)^\dagger \Theta_1$ and $\Theta_1 (\approx_b)^\dagger \Theta'$, by the transitivity of $\approx_b$,
we obtain  $\Delta (R\cdot\approx_b)^\dagger \Theta'$. Thus, the second condition is also satisfied.
\end{proof}

\begin{lemma} \label{branuptopro}
If $R$ is a branching bisimulation up to $\approx_b$, then $\approx_b\cdot R\cdot\approx_b$ is a branching bisimulation.
\end{lemma}
\begin{proof} Firstly, $\approx_b\cdot R\cdot\approx_b$ is symmetric as
$R$ is so being a branching bisimulation up to $\approx_b$. 
Next, since $\approx_b\subseteq\mathcal{B}(\approx_b)$ and by Lemma \ref{Lemma-upto} 
$R\cdot\approx_b\subseteq\mathcal{B}(R\cdot\approx_b)$,
then by Lemma \ref{Lemma-transitivity} $\approx_b\cdot R\approx_b\subseteq\mathcal{B}(\approx_b\cdot R\cdot\approx_b)$.
Hence
$\approx_b\cdot R\cdot\approx_b$ is a branching bisimulation.
\end{proof}

\begin{proposition}\label{branuptoisbran}
	If $R$ is a branching bisimulation up to $\approx_b$, then $R\subseteq=_b$.
\end{proposition}
\begin{proof}  Suppose $P R Q$, we need to show that $P=_bQ$.
We only need to show that if $P\arrow{\alpha}\Delta$, then there exists
$\Theta$ such that $Q\arrow{\alpha}\Theta$ and $\Delta(\approx_b)^\dagger\Theta$.
Now for $P\arrow{\alpha}\Delta$, since $R$ is a branching bisimulation up to $\approx_b$, there exists
$\Theta$ such that $\Delta(R\cdot\approx_b)^\dagger\Theta$, and then $\Delta(\approx_b)^\dagger\Theta$ follows 
from the following sequence of  inclusions:
$$R\cdot\approx_b\subseteq{\mathcal Id}_{\mathcal{P}}\cdot R\cdot\approx_b\subseteq\approx_b\cdot R\cdot\approx_b
\subseteq\approx_b
$$
The last inclusion is because by Lemma~\ref{branuptopro}  $\approx_b\cdot R\cdot\approx_b$ is a branching bisimulation.
\end{proof}

Next, we shall demonstrate that branching equality remains a congruence for recursively defined probabilistic process expressions. To this end, we first extend the notion of branching equality to probabilistic process expressions with variables via Definition~\ref{bisim-for-open-exp}.

\begin{definition}\label{bisim-for-open-exp}
Let $E,F\in{\mathcal E}$ be two expressions that contain variables $X_1,\ldots,X_n$ at most.
We say that $E$ and $F$ are branching equal, notation $E=_bF$, if for all $P_1,\ldots,P_n\in{\mathcal P}$
it holds that $$E\{P_1/X_1,\ldots,P_n/X_n\}=_bF\{P_1/X_1,\ldots,P_n/X_n\}.$$
\end{definition}

Finally, Proposition~\ref{branchingequalityiscongforrecursion} establishes that branching equality is a congruence for recursively defined process expressions.

\begin{proposition}\label{branchingequalityiscongforrecursion} Let $E_1,\ldots,E_n,F_1,\ldots,F_n\in{\mathcal E}$ contain variables $X_1,\ldots,X_n$ at most, and for $i=1,\ldots,n$
$A_i\stackrel{\rm def}{=}E_i\{A_1/X_1,\ldots,A_n/X_n\}$, 
$B_i\stackrel{\rm def}{=}F_i\{B_1/X_1,\ldots,B_n/X_n\}$. If $E_i=_bF_i$ for $i=1,\ldots,n$,
then $A_i=_bB_i$ for $i=1,\ldots,n$.
\end{proposition}
\begin{proof} For the sake of simplicity, we shall only consider the process expression using a single variable. Assume
$$\begin{array}{ccc}
E=_bF&
A\stackrel{\rm def}{=}E\{A/X\}&
B\stackrel{\rm def}{=}F\{B/X\}.
\end{array}
$$
We need to establish $A=_bB$. It suffices to show that $S\cup S^{-1}$ is a branching bisimulation up to $\approx_{b}$, where
$$S=\setof{(G\{A/X\},G\{B/X\})}{G\mbox{ contains at most variable }X}.
$$
To this end, it suffices to establish the following claim:
\begin{claim}\label{claim2}
 
  If $G\{A/X\}\arrow{\alpha}\Delta$,  then for some $\Theta\in{\sf Disc}(\mathcal P)$, we have  $G\{B/X\}\arrow{\alpha}\Theta$   and $\Delta(S\cdot\approx_b)^\dagger\Theta$.
 
\end{claim}

We prove Claim~\ref{claim2} by transition induction, i.e. induction on the depth of the inference for the transition $G\{A/X\}\arrow{\alpha}\Delta$.
We distinguish cases based on the form of $G$.\\

\textbf{1) $G\equiv X$.} In this case, we have $G\{A/X\}\equiv A$ and $G\{B/X\}\equiv B$. We need to prove that for any transition from A, e.g., $A\arrow{\alpha}\Delta$, there exists a $\Theta$ such that $B\arrow{\alpha}\Theta$ with $\Delta(S\cdot\approx_b)^\dagger\Theta$.
Since $A\stackrel{{\rm def}}{=} E\{A/X\}$,  then $A\arrow{\alpha}\Delta$ is inferred from 
$E\{A/X\}\arrow{\alpha}\Delta$,
hence, induction can be applied to conclude that $E\{B/X\}\arrow{\alpha}\Theta'$ with $\Delta(S\cdot\approx_b)^\dagger\Theta'$.
Now $E =_b F$, thus $E\{B/X\}=_bF\{B/X\}$, there exists a $\Theta$ such that $F\{B/X\}\arrow{\alpha}\Theta$ with $\Theta' (\approx_b)^\dagger \Theta$.
Since $B\stackrel{{\rm def}}{=} F\{B/X\}$, we get  $B \arrow{\alpha}\Theta$.
By proposition~{\ref{Prolift}},$\Delta(S\cdot\approx_b)^\dagger\Theta$ , as required.\\

\textbf{2) $G\equiv \alpha.\biguplus_{j\in J}\langle p_{j}\rangle G'_{j}$.}
In this case, $G\{A/X\}\equiv\alpha.\biguplus_{j\in J}\langle p_{j}\rangle G'_{j}\{A/X\} \arrow{\alpha} \Sigma_{j\in J}p_{j}\delta(G'_{j}\{A/X\})$.
Similarly, $G\{B/X\}\equiv\alpha.\biguplus_{j\in J}\langle p_{j}\rangle G'_{j}\{B/X\} \arrow{\alpha} \Sigma_{j\in J}p_{j}\delta(G'_{j}\{B/X\})$.
Since $(G'_{j}\{A/X\},G'_{j}\{B/X\})\in S$ for every $j\in J$, we get that $\Sigma_{j\in J}p_{j}\delta(G'_{j}\{A/X\}) S^ \dagger \Sigma_{j\in J}p_{j}\delta(G'_{j}\{B/X\})$ by Definition~{\ref{lifting-relation}}, as required.\\

\textbf{3) $G\equiv \sum_{i\in I}G'_{i}$.}
In this case, $G\{A/X\}\equiv\sum_{i\in I}G'_i\{A/X\}$. A transition from $G\{A/X\} $ is derivable if and only if $G'_n\{A/X\}\arrow{\alpha}\Delta$ for some $n\in I$.
Without loss of generality, assume the transition is due to $G'_1\{A/X\}$. By  hypothesis, there exists a $\Theta$ such that $G'_1\{B/X\}\arrow{\alpha}\Theta$ with $\Delta\, (S\cdot\approx_b)^\dagger\Theta$.
By the operational semantics, this ensures that $G\{B/X\}\equiv \sum_{i\in I}G'_i\{B/X\}\arrow{\alpha}\Theta$ with $\Delta\, ( S\cdot\approx_b)^\dagger\Theta$, as required.\\

\textbf{4) $G\equiv G_1|G_2$.} In this case, $G\{A/X\}\equiv G_1\{A/X\}|G_2\{A/X\}$.
There are three cases for the action $G\{A/X\}\arrow{\alpha}\Delta$, according to whether it arises from one or other component alone or from a communication.\\

\textbf{Case 1:} In this case, the $\alpha$ transition is due to $G_1\{A/X\}\arrow{\alpha}\Delta_1$ . Hence, by Definition~\ref{Opsemantics}, we have $\Delta=\Sigma_{P'\in\lceil\Delta_1\rceil}\Delta_1(P')\delta(P'|G_2\{A/X\})$ .
By the induction hypothesis, we have $G_1\{B/X\} \arrow{\alpha} \Delta_2$
with $\Delta_1 (S\,\cdot\approx_b)^\dagger \Delta_2$ . Then by definition~\ref{Opsemantics}, there exists $\Theta$ such that $G\{B/X\} \arrow{\alpha} \Theta$ where
$\Theta= \Sigma_{Q'\in\lceil\Delta_2\rceil}\Delta_2(Q')\delta(Q'|G_2\{B/X\})$.
Now we only need to show that
$\Delta(S\,\cdot\approx_b)^\dagger \Theta$ holds to establish the case.\\
Since $\Delta_1 (S\,\cdot\approx_b)^\dagger \Delta_2$, there exists some $\Delta_3$ such that $\Delta_1 S^\dagger \Delta_3$ and $\Delta_3 (\approx_b)^\dagger \Delta_2$. \\
Let
$\Theta'=\Sigma_{P''\in\mathcal{P}}\Delta_3(P'')\delta(P''|G_2\{B/X\})$. Then, to establish
$\Delta(S\,\cdot\approx_b)^\dagger \Theta$, it is sufficient to show that $\Delta S^\dagger\Theta'$ and $\Theta'(\approx_b)^\dagger \Theta$.
From part 3 of Proposition~\ref{Prop-DefofLifting} and $\Delta_1 S^\dagger \Delta_3$,
there exists a function $w_1:\mathcal{P}\times\mathcal{P}\rightarrow [0,1]$ such that:
\begin{enumerate}
	\item[a)] $P'S P''$ for all $w_1(P',P'')>0$,
	\item[b)] $\Delta_1(P')=\Sigma_{P''\in\mathcal{P}}w_1(P',P'')$ for all $P'\in \mathcal{P}$, and
	\item[c)] $\Delta_3(P'')=\Sigma_{P'\in\mathcal{P}}w_1(P',P'')$ for all $P''\in \mathcal{P}$.
\end{enumerate}
Then we have the following calculations:
$$\begin{array}{lclr}
	\Delta&=&\Sigma_{P'\in\lceil\Delta_1\rceil}\Delta_1(P')\delta(P'|G_2\{A/X\})&  \\
	&=&\Sigma_{P'\in\mathcal{P}}\Delta_1(P')\delta(P'|G_2\{A/X\})& \mbox{$\Delta_1(P')=0$
		for $P'\notin\lceil\Delta_1\rceil$}\\
	&=&
	\Sigma_{P'\in\mathcal{P}}(\Sigma_{P''\in\mathcal{P}}w_1(P',P''))\delta(P'|G_2\{A/X\})&\mbox{above b)}\\
	&=&
	\Sigma_{P'\in\mathcal{P},P''\in\mathcal{P}}w_1(P',P'')\delta(P'|G_2\{A/X\})\\
	&=&
	\Sigma_{P' S P''}w_1(P',P'')\delta(P'|G_2\{A/X\})&\mbox{$w_1(P',P'')=0$ for $(P',P'')\notin S $}\\
\end{array}
$$
Similarly, we can obtain
$\Theta'=\Sigma_{P' S P''}w_1(P',P'')\delta(P''|G_2\{B/X\})$.\\
Since $P'S P''$ and $G_2\{A/X\} S G_2\{B/X\}$, we have $(P'|G_2\{A/X\},P''|G_2\{B/X\})\in S$ , thus $\Delta S^\dagger\Theta'$.

On the other hand, from part 3 of Proposition~\ref{Prop-DefofLifting} and $\Delta_2 (\approx_b)^\dagger \Delta_3$,
there exists a function $w_2:\mathcal{P}\times\mathcal{P}\rightarrow [0,1]$ such that:
\begin{enumerate}
	\item[a)] $Q'\approx_b P''$ for all $w_2(Q',P'')>0$,
	\item[b)] $\Delta_2(Q')=\Sigma_{P''\in\mathcal{P}}w_2(Q',P'')$ for all $Q'\in \mathcal{P}$, and
	\item[c)] $\Delta_3(P'')=\Sigma_{Q'\in\mathcal{P}}w_2(Q',P'')$ for all $P''\in \mathcal{P}$.
\end{enumerate}
Then we have the following calculations:
$$\begin{array}{lclr}
	\Theta&=&\Sigma_{Q'\in\lceil\Delta_2\rceil}\Delta_2(Q')\delta(Q'|G_2\{B/X\})&  \\
	&=&\Sigma_{Q'\in\mathcal{P}}\Delta_2(Q')\delta({Q'|G_2\{B/X\}})& \mbox{$\Delta_2(Q')=0$
		for $Q'\notin\lceil\Delta_2\rceil$}\\
	&=&
	\Sigma_{Q'\in\mathcal{P}}(\Sigma_{P''\in\mathcal{P}}w_2(Q',P''))\delta({Q'|G_2\{B/X\}})&\mbox{above b)}\\
	&=&
	\Sigma_{Q'\in\mathcal{P},P''\in\mathcal{P}}w_2(Q',P'')\delta({Q'|G_2\{B/X\}})\\
	&=&
	\Sigma_{Q' \approx_b P''}w_2(Q',P'')\delta(Q'|G_2\{B/X\})&\mbox{$w_2(Q',P'')=0$ for $Q'\not\approx_b P''$}\\
\end{array}
$$
Similarly, we can obtain $\Theta'=\Sigma_{Q'\approx_b P''}w_2(Q',P'')\delta({P''|G_2\{B/X\}})$.
Since $Q'\approx_b P''$ , by congruence we have $(Q'|G_2\{B/X\})\approx_b(P''|G_2\{B/X\})$,  thus $\Theta(\approx_b)^\dagger\Theta'$, as required.
\\

\textbf{Case 2:} In contrast to Case 1, the $\alpha$ transition is due to $G_2\{A/X\}$.The proof of this case is similar to that of Case 1 and is thus omitted for brevity.\\

\textbf{Case 3:} We consider the case where $\alpha = \tau$, and we have $G_1\{A/X\}\arrow{l}\Delta_1$, and $G_2\{A/X\}\arrow{\bar{l}}\Delta_2$.
According to Definition~\ref{Opsemantics}, we have  $\Delta=\Sigma_{(U,V)\in\lceil\Delta_1\rceil\times\lceil\Delta_2\rceil}(\Delta_1(U)\cdot\Delta_2(V))\delta(U|V)$ .
Furthermore, based on the induction hypothesis, we have $G_1\{B/X\} \arrow{l} \Theta_1$
with $\Delta_1 (S\,\cdot\approx_b)^\dagger \Theta_1$, and $G_2\{B/X\} \arrow{\bar{l}} \Theta_2$
with $\Delta_2 (S\,\cdot\approx_b)^\dagger \Theta_2$.
Then by Definition~\ref{Opsemantics} again, there exists $\Theta$ such that $G\{B/X\} \arrow{\tau} \Theta$ where
$\Theta= \Sigma_{(U',V')\in\lceil\Theta_1\rceil\times\lceil\Theta_2\rceil}(\Theta_1(U')\cdot\Theta_2(V'))\delta(U'|V')$. Now we only need to show
$\Delta(S\,\cdot\approx_b)^\dagger \Theta$ to establish the case.\\
Since $\Delta_1 (S\,\cdot\approx_b)^\dagger \Theta_1$, there exists some $\Delta_1^{'}$ such that $\Delta_1 S^\dagger \Delta_1^{'} (\approx_b)^\dagger \Theta_1$ . In the same manner, since $\Delta_2 (S\,\cdot\approx_b)^\dagger \Theta_2$, there exists some $\Delta_2^{'}$ such that $\Delta_2 S^\dagger \Delta_2^{'} (\approx_b)^\dagger \Theta_2$ .\\
Let $\Delta'=\Sigma_{(U'',V'')\in\lceil\Delta_1^{'}\rceil\times\lceil\Delta_2^{'}\rceil}(\Delta_1^{'}(U'')\cdot\Delta_2^{'}(V''))\delta(U''|V'')$, then to establish $\Delta (S\cdot\approx_b)^\dagger \Theta $
it is sufficient to show that $\Delta S^\dagger \Delta'$
and $\Delta' (\approx_b)^\dagger \Theta$.
From part 3 of Proposition~\ref{Prop-DefofLifting} and $\Delta_1 S^\dagger \Delta_1^{'}$,
there exists a function $w_1:\mathcal{P}\times\mathcal{P}\rightarrow [0,1]$ such that:
\begin{enumerate}
	\item[a)] $U S U''$ for all $w_1(U,U'')>0$,
	\item[b)] $\Delta_1(U)=\Sigma_{U''\in\mathcal{P}}w_1(U,U'')$ for all $U\in \mathcal{P}$, and
	\item[c)] $\Delta_1^{'}(U'')=\Sigma_{U\in\mathcal{P}}w_1(U,U'')$ for all $U''\in \mathcal{P}$.
\end{enumerate}
Following the same reasoning, from part 3 of Proposition~\ref{Prop-DefofLifting} and $\Delta_2 S^\dagger \Delta_2^{'}$,
there exists a function $w_2:\mathcal{P}\times\mathcal{P}\rightarrow [0,1]$ such that:
\begin{enumerate}
	\item[d)] $V S V''$ for all $w_2(V,V'')>0$,
	\item[e)] $\Delta_2(V)=\Sigma_{V''\in\mathcal{P}}w_2(V,V'')$ for all $V\in \mathcal{P}$, and
	\item[f)] $\Delta_2^{'}(V'')=\Sigma_{V\in\mathcal{P}}w_2(V,V'')$ for all $V''\in \mathcal{P}$.
\end{enumerate}	
Based on these definitions and conditions, we can proceed with the following calculations:
$$\begin{array}{lclr}
	\Delta&=&\Sigma_{(U,V)\in\lceil\Delta_1\rceil\times\lceil\Delta_2\rceil}(\Delta_1(U)\cdot\Delta_2(V))\delta(U|V)&  \\
	&=&\Sigma_{U\in\lceil\Delta_1\rceil}\Delta_1(U)\Sigma_{V\in\lceil\Delta_2\rceil}\Delta_2(V)\delta(U|V)&  \\
	&=&\Sigma_{U\in\mathcal{P}}\Delta_1(U)\Sigma_{V\in\lceil\Delta_2\rceil}\Delta_2(V)\delta(U|V)& \mbox{$\Delta_1(U)=0$ for $U\notin\lceil\Delta_1\rceil$}\\
	&=&\Sigma_{U\in\mathcal{P}}\Delta_1(U)\Sigma_{V\in\mathcal{P}}\Delta_2(V)\delta(U|V)& \mbox{$\Delta_2(V)=0$	for $V\notin\lceil\Delta_2\rceil$}\\
	&=&\Sigma_{U\in\mathcal{P}}(\Sigma_{U''\in\mathcal{P}}w_1(U,U''))\Sigma_{V\in\mathcal{P}}(\Sigma_{V''\in\mathcal{P}}w_2(V,V''))\delta(U|V)&\mbox{above b),e)}\\
	&=&\Sigma_{U\in\mathcal{P},U''\in\mathcal{P}}w_1(U,U'')\Sigma_{V\in\mathcal{P},V''\in\mathcal{P}}w_2(V,V'')\delta(U|V)\\
	&=&\Sigma_{U S U''}w_1(U,U'')\Sigma_{V\in\mathcal{P},V''\in\mathcal{P}}w_2(V,V'')\delta(U|V)&\mbox{$w_1(U,U'')=0$ for $(U,U'')\notin S $}\\
	&=&\Sigma_{U S U''}w_1(U,U'')\Sigma_{V S V''}w_2(V,V'')\delta(U|V)&\mbox{$w_2(V,V'')=0$ for $(V,V'')\notin S $}\\
\end{array}
$$
Similarly, we can obtain $\Delta'=\Sigma_{U S U''}w_1(U,U'')\Sigma_{V S V''}w_2(V,V'')\delta(U''|V'')$. We know that $(U|V,U''|V'')\in S$. Therefore, it follows that $\Delta S^\dagger \Delta'$.\\
Since $\Theta_1 (\approx_b)^\dagger \Delta_1^{'}$, from part 3 of Proposition~\ref{Prop-DefofLifting},
there exists a function $w_3:\mathcal{P}\times\mathcal{P}\rightarrow [0,1]$ such that
\begin{enumerate}
	\item[g)] $U' \approx_b U''$ for all $w_3(U',U'')>0$,
	\item[h)] $\Delta_1^{'}(U'')=\Sigma_{U'\in\mathcal{P}}w_3(U',U'')$ for all $U''\in \mathcal{P}$, and
	\item[i)] $\Theta_1(U')=\Sigma_{U''\in\mathcal{P}}w_3(U',U'')$ for all $U'\in \mathcal{P}$.
\end{enumerate}
Similarly, since $\Theta_2 (\approx_b)^\dagger \Delta_2^{'}$, from part 3 of Proposition~\ref{Prop-DefofLifting},there exists a function $w_4:\mathcal{P}\times\mathcal{P}\rightarrow [0,1]$ such that
\begin{enumerate}
	\item[j)] $V' \approx_b V''$ for all $w_4(V',V'')>0$,
	\item[k)] $\Delta_2^{'}(V'')=\Sigma_{V'\in\mathcal{P}}w_4(V',V'')$ for all $V''\in \mathcal{P}$, and
	\item[l)] $\Theta_2(V')=\Sigma_{V''\in\mathcal{P}}w_4(V',V'')$ for all $V'\in \mathcal{P}$.
\end{enumerate}
Next, We proceed with the following calculations:
$$\begin{array}{lclr}		
	\Theta&=&\Sigma_{(U',V')\in\lceil\Theta_1\rceil\times\lceil\Theta_2\rceil}(\Theta_1(U')\cdot\Theta_2(V'))\delta(U|V)&  \\
	&=&\Sigma_{U'\in\lceil\Theta_1\rceil}\Theta_1(U')\Sigma_{V'\in\lceil\Theta_2\rceil}\Theta_2(V')\delta(U'|V')&  \\
	&=&\Sigma_{U'\in\mathcal{P}}\Theta_1(U')\Sigma_{V'\in\lceil\Theta_2\rceil}\Theta_2(V')\delta(U'|V')& \mbox{$\Theta_1(U')=0$ for $U'\notin\lceil\Theta_1\rceil$}  \\
	&=&\Sigma_{U'\in\mathcal{P}}\Theta_1(U')\Sigma_{V'\in\mathcal{P}}\Theta_2(V')\delta(U'|V')& \mbox{$\Theta_2(V')=0$ for $V'\notin\lceil\Theta_2\rceil$}  \\
	&=&\mbox{\small$\displaystyle
		\Sigma_{U'\in\mathcal{P}}(\Sigma_{U''\in\mathcal{P}}w_3(U',U''))\Sigma_{V'\in\mathcal{P}}(\Sigma_{V''\in\mathcal{P}}w_4(V,V''))\delta(U'|V')$}&\mbox{above i),l)}\\
	&=&\Sigma_{U'\in\mathcal{P},U''\in\mathcal{P}}w_3(U',U'')\Sigma_{V'\in\mathcal{P},V''\in\mathcal{P}}w_4(V',V'')\delta(U'|V')\\
	&=&\Sigma_{U' \approx_b U''}w_3(U',U'')\Sigma_{V'\in\mathcal{P},V''\in\mathcal{P}}w_4(V',V'')\delta(U'|V')&\mbox{$w_3(U',U'')=0$ for $ U'\not\approx_b U''  $}\\
	&=&\Sigma_{U' \approx_b U''}w_3(U',U'')\Sigma_{V' \approx_b V''}w_4(V',V'')\delta(U'|V')&\mbox{$w_4(V',V'')=0$ for $V' \not\approx_b V'' $}\\
\end{array}
$$
Similarly, we can obtain $\Delta'=\Sigma_{U' \approx_b U''}w_3(U',U'')\Sigma_{V' \approx_b V''}w_4(V',V'') \delta(U''|V'')$.
Given that $U' \approx_b U''$ and $V' \approx_b V''$, it follows by congruence that $(U'|V')\approx_b(U''|V'')$. Consequently, we establish that $\Theta (\approx_b)^\dagger \Delta'$. This completes the proof.\\

\textbf{5) $G\equiv G_1\backslash L$.} In this case, $G\{A/X\} \equiv G_1\{A/X\}\backslash L$.
Assume $G\{A/X\} \arrow{\alpha} \Delta$, then there exists a transition $G_1\{A/X\} \arrow{\alpha} \Delta_1$
with $\alpha \notin L\cup\overline{L}$ such that
$\Delta= \Sigma_{P'\in\lceil\Delta_1\rceil}\Delta_1(P')\delta({P'\backslash L})$.
By the induction hypothesis, we have $G_1\{B/X\} \arrow{\alpha} \Delta_2$
with $\Delta_1 (S\,\cdot\approx_b)^\dagger \Delta_2$ . Then by Definition~\ref{Opsemantics}, there exists $\Theta$ such that $G\{B/X\} \arrow{\alpha} \Theta$ where
$\Theta= \Sigma_{Q'\in\lceil\Delta_2\rceil}\Delta_2(Q')\delta({Q'\backslash L})$. Now we only need to show
$\Delta(S\,\cdot\approx_b)^\dagger \Theta$ to establish the case. 
Since $\Delta_1 (S\,\cdot\approx_b)^\dagger \Delta_2$, there exists some $\Delta_3$ such that $\Delta_1 S^\dagger \Delta_3 (\approx_b)^\dagger \Delta_2$, and let
$\Theta'=\Sigma_{P''\in\mathcal{P}}\Delta_3(P'')\delta(P''\backslash L)$, then to establish
$\Delta(S\,\cdot\approx_b)^\dagger \Theta$ it is sufficient to show that $\Delta S^\dagger\Theta'(\approx_b)^\dagger \Theta$.
From 3 of Proposition~\ref{Prop-DefofLifting} and $\Delta_1 S^\dagger \Delta_3$,
there exists a function $w_1:\mathcal{P}\times\mathcal{P}\rightarrow [0,1]$ such that
\begin{enumerate}
	\item[a)] $P'S P''$ for all $w_1(P',P'')>0$,
	\item[b)] $\Delta_1(P')=\Sigma_{P''\in\mathcal{P}}w_1(P',P'')$ for all $P'\in \mathcal{P}$, and
	\item[c)] $\Delta_3(P'')=\Sigma_{P'\in\mathcal{P}}w_1(P',P'')$ for all $P''\in \mathcal{P}$.
\end{enumerate}
Then we have the following calculations:
$$\begin{array}{lclr}
	\Delta&=&\Sigma_{P'\in\lceil\Delta_1\rceil}\Delta_1(P')\delta(P'\backslash L)&  \\
	&=&\Sigma_{P'\in\mathcal{P}}\Delta_1(P')\delta({P'\backslash L})& \mbox{$\Delta_1(P')=0$
		for $P'\notin\lceil\Delta_1\rceil$}\\
	&=&
	\Sigma_{P'\in\mathcal{P}}(\Sigma_{P''\in\mathcal{P}}w_1(P',P''))\delta({P'\backslash L})&\mbox{above b)}\\
	&=&
	\Sigma_{P'\in\mathcal{P},P''\in\mathcal{P}}w_1(P',P'')\delta({P'\backslash L})\\
	&=&
	\Sigma_{P' S P''}w_1(P',P'')\delta({P'\backslash L})&\mbox{$w_1(P',P'')=0$ for $(P',P'')\notin S $}\\
\end{array}
$$
Similarly, we can obtain
$\Theta'=\Sigma_{P' S P''}w_1(P',P'')\delta({P''\backslash L})$, thus $\Delta S^\dagger\Theta'$.

On the other hand, from 3 of Proposition~\ref{Prop-DefofLifting} and $\Delta_2 (\approx_b)^\dagger \Delta_3$,
there exists a function $w_2:\mathcal{P}\times\mathcal{P}\rightarrow [0,1]$ such that
\begin{enumerate}
	\item[a)] $Q'\approx_b P''$ for all $w_2(Q',P'')>0$,
	\item[b)] $\Delta_2(Q')=\Sigma_{P''\in\mathcal{P}}w_2(Q',P'')$ for all $Q'\in \mathcal{P}$, and
	\item[c)] $\Delta_3(P'')=\Sigma_{Q'\in\mathcal{P}}w_2(Q',P'')$ for all $P''\in \mathcal{P}$.
\end{enumerate}
Then we have the following calculations:
$$\begin{array}{lclr}
	\Theta&=&\Sigma_{Q'\in\lceil\Delta_2\rceil}\Delta_2(Q')\delta(Q'\backslash L)&  \\
	&=&\Sigma_{Q'\in\mathcal{P}}\Delta_2(Q')\delta({Q'\backslash L})& \mbox{$\Delta_2(Q')=0$
		for $Q'\notin\lceil\Delta_2\rceil$}\\
	&=&
	\Sigma_{Q'\in\mathcal{P}}(\Sigma_{P''\in\mathcal{P}}w_2(Q',P''))\delta({Q'\backslash L})&\mbox{above b)}\\
	&=&
	\Sigma_{Q'\in\mathcal{P},P''\in\mathcal{P}}w_2(Q',P'')\delta({Q'\backslash L})\\
	&=&
	\Sigma_{Q' \approx_b P''}w_2(Q',P'')\delta({Q'\backslash L})&\mbox{$w_2(Q',P'')=0$ for $Q'\not\approx_b P'' $}\\
\end{array}
$$
Similarly, we can obtain $\Theta'=\Sigma_{Q'\approx_b P''}w_2(Q',P'')\delta({P''\backslash L})$, thus $\Theta(\approx_b)^\dagger\Theta'$.
Hence $\Delta S^\dagger \Theta' (\approx_b)^\dagger \Theta$ by Proposition~\ref{Prolift}, that is $\Delta\, ( S\cdot\approx_b)^\dagger\Theta$, as required.\\
\textbf{6) $G\equiv G_1[f]$.} The proof of this follows the same lines as that of 5) above, and thus we omit the details.
\end{proof}
\section{Related Works}\label{sec8}

Probabilistic process algebras have been investigated for decades, giving rise to models such as Strictly-alternating systems~\cite{HanssonJ90}, Simple probabilistic automata~\cite{SegalaL94}, Concurrent labelled Markov chains~\cite{PhilippouLS00} and Probabilistic transition systems (PTS) ~\cite{Jonsson01,Deng18}.
Among them, models that combine nondeterminism with probability have become the mainstream.
The present study adopts probabilistic transition systems, which conservatively extend ordinary labelled transition systems and thus provide a compatible foundation for our work.

 
In classical CCS theory, congruence proofs, especially those concerning recursion, are already highly involved; in the probabilistic setting, they are regarded as notoriously difficult challenges~\cite{Fu21,GlabbeekGV19,GlabbeekGV25}. We extend the up-to technique to the probabilistic realm, establishing that 
branching equality is a congruence for the entire probabilistic CCS language, including parallel composition and recursion. 

Building on the strictly alternating model~\cite{HanssonJ90}, Andova and Willemse ~\cite{AndovaW06} introduced a state-based branching bisimulation relation.  In this model, states are classified as either nondeterministic or probabilistic, and transitions are either action or probabilistic. However, their branching bisimulation failed to be a congruence for parallel composition~\cite{TrckaG08}. Subsequently, Andova et al.~\cite{AndovaGT12} proposed a revised notion of branching bisimulation within the same model. They defined each branching bisimulation relation directly as an equivalence and proved that rooted branching bisimilarity is a congruence for the operators of $pACP_\tau$, an extension of the process algebra ACP\cite{Baeten90}. Notably, their approach and model differ significantly from those adopted in this paper. Requiring every branching bisimulation to be an equivalence is overly restrictive; we do not impose this condition to avoid the unnecessary overhead it imposes. 

Following Segala's probabilistic automata framework\cite{Segala01}, Glabbeek et al.~\cite{GlabbeekGV19}, with a recent revision in \cite{GlabbeekGV25}, proposed a mixed non-deterministic/probabilistic process algebra, whose syntax contains only 0 (inaction), action prefix, and non-deterministic choice; standard CCS operators such as parallel composition and restriction, as well as recursion, are absent. They defined a dedicated transition relation that connects probabilistic processes to probability distributions. 
Building on this model, they introduced the notions of strong probabilistic bisimilarity and branching probabilistic bisimilarity, both of which they prove to be equivalence relations. Due to the fact that their equivalence is basically defined on distributions which unavoidably using combined transition, their branching equivalence is
coarser than the corresponding equivalence that we presented in this paper.
Moreover, they showed that strong bisimilarity and rooted branching bisimilarity are congruences for the operators present in their language. A notable limitation is that their congruence results do not cover the parallel operator and recursion. 

It should also be noted that branching bisimilarity became popular shortly after~\cite{GlabbeekW89}, and branching bisimilarity has since been regarded by many researchers as an equivalence relation ``by nature''.
However, Basten~\cite{Basten96} pointed out that the composition of two branching bisimulation relations is not necessarily a branching bisimulation relation, so an equivalence proof for branching bisimilarity is not so trivial as one would hope.
He overcame this obstacle with an improved definition---semi-branching bisimulation---and a stuttering lemma, finally showing that branching bisimilarity is an equivalence relation.
Castiglioni and Tini~\cite{CastiglioniT20} were the first to lift this approach to the probabilistic setting: they introduced the probabilistic analogue of semi-branching bisimulation and Stuttering Lemma, proving that probabilistic branching bisimilarity is indeed an equivalence relation over divergence-free PTSs. Their definition of branching bisimulation
(semi-branching bisimulation) result in a coarser equivalence relation than the one presented in this paper. We have also being very careful in ensuring the the finall
relation that we define is an equivalence relation.

\section{Conlusion}\label{sec9}

In this paper we have extended the process language CCS with probability, and introduced a new notion of branching bisimilation for probabilistic processes. 
By forgetting probabilistic information in internal transitions and requiring matching of individual transitions, the
new notion captures more accurately the branching structure of probabilistic processes than existing notions of branching bisimulation which use combined
probabilistic transitions. We proved that the rooted version, branching equality, is a congruence on the entire extended language, while doing
so introduced an up to technique to deal with recursion.

\bibliographystyle{unsrtnat}
\bibliography{references}  






\end{document}